\documentclass[peerreviewca,draftclsnofoot,12pt,onecolumn]{IEEEtran}
\usepackage{ifpdf}              
\usepackage{hyperref}           
\usepackage{cite}               
\usepackage[cmex10]{amsmath}    
\usepackage{array}              
\usepackage{amsfonts}
\usepackage{amssymb}
\usepackage{xcolor}
\usepackage{lipsum}
\usepackage{mathtools}
\usepackage{stmaryrd}

\ifpdf
 \usepackage[pdftex]{graphicx}
 \usepackage{epstopdf}
 \DeclareGraphicsExtensions{.eps,.jpg,.png,.pdf}
\else
 \usepackage[dvips]{graphicx}
 \DeclareGraphicsExtensions{.eps}
\fi

\newtheorem{theorem}{Theorem}
\newtheorem{lemma}[theorem]{Lemma}

\newtheorem{corollary}[theorem]{Corollary}

\newenvironment{proof}[1][Proof]{\begin{trivlist}\item[\hskip \labelsep {\bfseries #1}]}{\end{trivlist}}
\newenvironment{definition}[1][Definition]{\begin{trivlist}\item[\hskip \labelsep {\bfseries #1}]}{\end{trivlist}}
\newenvironment{example}[1][Example]{\begin{trivlist}\item[\hskip \labelsep {\bfseries #1}]}{\end{trivlist}}

\newcommand{\qed}{\nobreak \ifvmode \relax \else
      \ifdim\lastskip<1.5em \hskip-\lastskip
      \hskip1.5em plus0em minus0.5em \fi \nobreak
      \vrule height0.75em width0.5em depth0.25em\fi}

\newcommand{\bibpath}{.} 

\newcommand{\E}[1]{\mathbb{E} \left\{#1 \right\}}

\newcommand{\Red}[1]{I_\cap(#1)}
\newcommand{\Un}[1]{I_\text{priv}(#1)}

\newcommand{\Syn}[1]{I_\text{S}(#1)}
\newcommand{\dr}[2]{\mathcal{R}(#1\!\shortrightarrow\! #2)}
\newcommand{\du}[2]{\mathcal{U}(#1\!\shortrightarrow\! #2)}

\newcommand{\argmax}{\operatornamewithlimits{argmax}}

\ifCLASSOPTIONpeerreview

\else

\fi

\title{Understanding interdependency through complex information sharing}

\ifCLASSOPTIONpeerreview  
    \author{\IEEEauthorblockN{Fernando Rosas$^1$, Vasilis Ntranos$^2$, Christopher J. Ellison$^3$,\\ Sofie Pollin$^1$ and Marian Verhelst$^1$}\\
    \IEEEauthorblockA{$^1$ Departement Elektrotechniek, KU Leuven}
                                  {$^2$ Department of Electrical Engineering \& Computer Sciences, UC Berkeley}\\
                                  {$^3$ Center for Complexity and Collective Computation, University of Wisconsin-Madison}
}

  \else 
    \author{\IEEEauthorblockN{Fernando Rosas$^1$, Vasilis Ntranos$^2$, Christopher J. Ellison$^3$,\\ Sofie Pollin$^1$ and Marian Verhelst$^1$}\\
    \IEEEauthorblockA{$^1$ Departement Elektrotechniek, KU Leuven} \\
    \IEEEauthorblockA{$^2$ EECS, UC Berkeley}\\
    \IEEEauthorblockA{$^3$ Center for Complexity and Collective Computation, University of Wisconsin-Madison}\\

}

\fi

\hypersetup{
  breaklinks  = {true},
  colorlinks  = {true},
  linktocpage = {false},
  linkcolor   = {blue},
  citecolor   = {green!60!black},
  urlcolor    = {red!60!black},
  plainpages  = {false},
  pageanchor  = {true}
}

\begin{document}

\maketitle

\begin{abstract}
The interactions between three or more random variables are often nontrivial, poorly understood, and yet, are paramount for future advances in fields such as network information theory, neuroscience, genetics and many others. In this work, we propose to analyze these interactions as different modes of information sharing. Towards this end, we introduce a novel axiomatic framework for decomposing the joint entropy, which characterizes the various ways in which random variables can share information.
The key contribution of our framework is to distinguish between interdependencies where the information is shared redundantly, and synergistic interdependencies where the sharing structure exists in the whole but not between the parts. We show that our axioms determine unique formulas for all the terms of the proposed decomposition for a number of cases of interest. Moreover, we show how these results can be applied to several network information theory problems, providing a more intuitive understanding of their fundamental limits.
\end{abstract}

\ifCLASSOPTIONpeerreview
\begin{center}
Corresponding author: Fernando Rosas\\\texttt{fernando.rosas@esat.kuleuven.be}.\\
\end{center}
\fi
\IEEEpeerreviewmaketitle


\section{Introduction} \label{form}

Interdependence is a key concept for understanding the rich structures that can be exhibited by biological, economical and social systems~\cite{kaneko2006life,perrings2005economy}. Although this phenomenon lies in the heart of our modern interconnected world, there is still no solid quantitative framework for analyzing complex interdependences, this being crucial for future advances in a number of disciplines. In neuroscience, researchers desire to identify how various neurons affect an organism's overall behavior, asking to what extent the different neurons are providing redundant or synergistic signals~\cite{martignon2000}. In genetics, the interactions and roles of multiple genes with respect to phenotypic phenomena are studied, e.g. by comparing results from single and double knockout experiments~\cite{deutscher2008can}. In graph and network theory, researchers are looking for measures of the information encoded in node interactions in order to quantify the complexity of the network~\cite{anand2009entropy}. In communication theory, sensor networks usually generate strongly correlated data~\cite{gastpar2006sensing}; a haphazard design might not account for these interdependencies and, undesirably, will process and transmit redundant information across the network degrading the efficiency of the system.

The dependencies that can exist between two variables have been extensively studied, generating a variety of techniques that range from statistical inference~\cite{Casella2002} to information theory~\cite{Cover1991}. Most of these approaches require that one differentiate the role of the variables, e.g. between a \textit{target} and \textit{predictor}. However, the extension of these approaches to three or more variables is not straightforward, as a binary splitting is, in general, not enough to characterize the rich interplay that can exist between variables. Moreover, the development of more adequate frameworks has been difficult as most of our theoretical tools are rooted in sequential reasoning, which is adept at representing linear flows of influences but not as well-suited for describing distributed systems or complex interdependencies~\cite{senge2008necessary}.

In this work, we propose to understand interdependencies between variables as {\em information sharing}. In the case of two variables, the portion of the variability that can be predicted corresponds to information that target and predictor have in common. Following this intuition, we present a framework that decomposes the total information of a distribution according to how it is shared among its variables. Our framework is novel in combining the hierarchical decomposition of higher-order interactions, as developed in \cite{Amari2001}, with the notion of synergistic information, as proposed in \cite{williams2010}. In contrast to \cite{Amari2001}, we study the information that exists in the system itself without comparing it with other related distributions. In contrast to \cite{williams2010}, we analyze the joint entropy instead of the mutual information, looking for symmetric properties of the system.

One important contribution of this paper is to distinguish \textit{shared information} from \textit{predictability}. Predictability is a concept that requires a bipartite system divided into predictors and targets. As different splittings of the same system often yield different conclusions, we see predictability as a directed notion that strongly depends on one's ``point of view''. In contrast, we see shared information as a property of the system itself, which does not require differentiated roles between its components. Although it is not possible in general to find an unique measure of predictability, we show that the shared information can be uniquely defined for a number of interesting scenarios.

Additionally, our framework provides new insight to various problems of network information theory. Interestingly, many of the problems of network information theory that have been solved are related to systems which present a simple structure in terms of shared information and synergies, while most of the open problems possess a more complex mixture of them.

The rest of this article is structured as follows. First, Section~\ref{sec:2} introduces the notions of hierarchical decomposition of dependencies and synergistic information, reviewing the state-of-the-art and providing the necessary background for the unfamiliar reader. Section~\ref{sec:3} presents our axiomatic decomposition for the joint entropy, focusing on the fundamental case of three random variables. Then, we illustrate the application of our framework for various cases of interest: pairwise independent variables in Section~\ref{sec:indep}, pairwise maximum entropy distributions and Markov chains in Section~\ref{sec:V}, and multivariate Gaussians in~\ref{sec:gaussian}. After that, Section~\ref{sec:4} presents a first application of this framework in settings of fundamental importance for network information theory. Finally, Section~\ref{sec:conclusions} summarizes our main conclusions.


\section{Preliminaries and state of the art}\label{sec:2}

One way of analyzing the interactions between the random variables $\mathbf{X}=(X_1,\dots,X_N)$ is to study the  properties of the correlation matrix $\mathcal{R}_{\mathbf{X}} = \E{ \mathbf{X} \mathbf{X}^t}$. However, this approach only captures linear relationships and hence the picture provided by $\mathcal{R}_{\mathbf{X}}$ is incomplete. Another possibility is to study the matrix $\mathcal{I}_{\mathbf{X}} = [ I(X_i; X_j) ]_{i,j}$ of mutual information terms. This matrix captures the existence of both linear and nonlinear dependencies~\cite{li1990mutual}, but its scope is restricted to pairwise relationships and thus misses all higher-order structure. To see an example of how this can happen, consider two independent fair coins $X_1$ and $X_2$ and let $X_3 \coloneqq X_1 \oplus X_2$ be the output of an \texttt{XOR} logic gate. The mutual information matrix $\mathcal{I}_{\mathbf{X}}$ has all its off-diagonal elements equal to zero, making it indistinguishable from an alternative situation where $X_3$ is just another independent fair coin.

For the case of $\mathcal{R}_{\mathbf{X}}$, a possible next step would be to consider higher-order moment matrices, such as co-skewness and co-kurtosis. We seek their information-theoretic analogs, which complement the description provided by $\mathcal{I}_{\mathbf{X}}$. One method of doing this is by studying the information contained in marginal distributions of increasingly larger sizes; this approach is presented in Section~\ref{sec:2.1}. Other methods try to provide a direct representation of the information that is shared between the  random variables; they are discussed in Sections~\ref{sec:2b}, \ref{sec:2c} and \ref{sec:2d}.

\subsection{Negentropy and total correlation}
\label{sec:2.1}

When the random variables that compose a system are independent, their joint distribution is given by the product of their marginal distributions. In this case, the marginals contain all that is to be learned about the statistics of the entire system. For an arbitrary joint probability density function (p.d.f.), knowing the single variable marginal distributions is not enough to capture all there is to know about the statistics of the system.

To quantify this idea, let us consider $N$ discrete random variables $\mathbf{X}=(X_1,\dots,X_N)$ with joint p.d.f. $p_{\mathbf{X}}$, where each $X_j$ takes values in a finite set with cardinality ${\Omega}_{j}$. The maximal amount of information that could be stored in any such system is $H^{(0)} = \sum_j \log \Omega_j$, which corresponds to the entropy of the p.d.f. $p_{\mathbf{U}} \coloneqq \prod_{j}\overline p_{X_{j}}$, where $\overline p_{X_{j}}(x) = 1/\Omega_j$ is the uniform distribution for each random variable $X_j$. On the other hand, the  joint entropy $H(\mathbf{X})$ with respect to the true distribution $p_{\mathbf{X}}$ measures the actual uncertainty that the system possesses. Therefore, the difference
\begin{equation}
\mathcal{N}(\mathbf{X}) \coloneqq H^{(0)} - H(\mathbf{X})
\end{equation}
corresponds to the decrease of the uncertainty about the system that occurs when one learns its p.d.f. -- i.e. the information about the system that is contained in its statistics. This quantity is known as \emph{negentropy}~\cite{brillouin1953negentropy}, and can also be computed~as
\begin{align}\label{eq:neg}
\mathcal{N}(X_1,\dots,X_N) &=\sum_j [ \log\Omega_j - H(X_j)] + \left(\sum_j H(X_j) - H(\mathbf{X})\right) \\
&= D\left(\prod_j p_{X_j} \;\Big|\Big| \;p_\bold{U}\right) + D\left(p_\bold{X}\;\Big|\Big| \;\prod_j p_{X_j}\right) \label{eq:neg}
\enspace,
\end{align}
where $p_{X_j}$ is the marginal of the variable $X_j$ and $D(\cdot || \cdot)$ is the Kullback-Leibler divergence. In this way, \eqref{eq:neg} decomposes the negentropy into a term that corresponds to the information given by simple marginals and a term that involves higher-order marginals. The second term is known as the \textit{total correlation} (TC)~\cite{watanabe1960information} (also known as \textit{multi-information}~\cite{studeny1998multiinformation}), which is equal to the mutual information for the case of $N=2$. Because of this, the $\text{TC}$ has been suggested as an extension of the notion of mutual information for multiple variables.

An elegant framework for decomposing the TC can be found in \cite{Amari2001} (for an equivalent formulation that do not rely on information geometry c.f. \cite{schneidman2003network}). Let us call $k$-marginals the distributions that are obtained by marginalizing the joint p.d.f. over $N-k$ variables. Note that the $k$-marginals provide a more detailed description of the system than the $(k-1)$-marginals, as the latter can be directly computed from the former by marginalizing the corresponding variables. In the case where only the $1$-marginals are known, the simplest guess for the joint distribution is $\tilde{p}^{\scriptscriptstyle\,(1)}_{\small\mathbf{X}} = \prod_j p_{X_j}$.
One way of generalizing this for the case where the $k$-marginals are known is by using the \textit{maximum entropy principle} \cite{jaynes2003}, which suggests to choose the distribution that maximizes the joint entropy while satisfying the constrains given by the partial ($k$-marginal) knowledge. Let us denote by $\tilde{p}^{\scriptscriptstyle\,(k)}_{\mathbf{X}}$ the p.d.f. which achieves the maximum entropy while being consistent with all the $k$-marginals, and let $H^{(k)}=H(\{\tilde{p}^{\scriptscriptstyle\,(k)}_{\mathbf{X}}\})$ denote its entropy. Note that $H^{(k)} \geq H^{(k+1)} $, since the number of constrains that are involved in the maximization process that generates $H^{(k)}$ increases with $k$. It can therefore be shown that the following generalized Pythagorean relationship holds for the total correlation:
\begin{equation}\label{eq:amari}
\text{TC} = H^{(1)} - H(\mathbf{X}) = \sum_{k=2}^{N} \left[ H^{(k-1)} - H^{(k)} \right]  = \sum_{k=2}^{N} D\left(\tilde{p}^{\scriptscriptstyle\,(k)}||\tilde{p}^{\scriptscriptstyle\,(k-1)}\right) \coloneqq \sum_{k=2}^{N} \Delta H^{(k)}
\enspace.
\end{equation}
Above, $\Delta H^{(k)}\geq 0$ measures the additional information that is provided by the $k$-marginals that was not contained in the description of the system given by the $(k-1)$-marginals. In general, the information that is located in terms with higher values of $k$ is due to dependencies between groups of variables that cannot be reduced to combinations of dependencies between smaller groups.

It has been observed that in many practical scenarios most of the $\text{TC}$ of the measured data is provided by the lower marginals. It can be shown that percentage of the $\text{TC}$ that is lost by considering only the $k_0$-order marginals is given by
 \begin{equation}
\frac{\text{TC} - \sum_{k=1}^{k_0} \Delta H^{(k)}}{\text{TC}} = \frac{1}{\text{TC}}\sum_{k=k_0+1}^N \Delta H^{(k)}= \frac{1}{\text{TC}} D\left(p_{\mathbf{X}}||\tilde{p}^{(k_0)}_{\mathbf{X}}\right)
\enspace.
\end{equation}
This quantity is small if there exists a value of $k_0$ such that $\tilde{p}^{(k_0)}_{\mathbf{X}}$ provides an accurate approximation for the joint p.d.f. of the system. Interestingly, it has been shown that pairwise maximum entropy models (i.e. $k_0=2$) can provide an accurate description of the statistics of many biological systems~\cite{schneidman2006weak,Roudi2009,bialek2012statistical,merchan2015sufficiency} and also some social organizations \cite{Daniels2012,lee2013statistical}.

\subsection{Internal and external decompositions}
\label{sec:2b}

An alternative approach to study the interdependencies between many random variables is to analyze the ways in which they share information. This can be done by decomposing the joint entropy of the system. For the case of two variables, the joint entropy can be decomposed as 
\begin{equation}
H(X_1,X_2) =   I(X_1;X_2)  + H(X_1|X_2) + H(X_2|X_1)
\enspace,
\end{equation}
suggesting that it can be divided into shared information, $I(X_1;X_2)$, and into terms which represent information that is exclusively located in a single variable, i.e., $H(X_1|X_2)$ for $X_{1}$ and $H(X_2|X_1)$ for $X_{2}$.

In systems with more than two variables, one can compute the total information that is exclusively located in one variable as $H_{(1)}\coloneqq \sum_j H(X_j|\mathbf{X}_j^c)$\footnote{The superscripts and subscripts are used to reflect that $H^{(1)} \geq H(\mathbf{X}) \geq H_{(1)}$.}, where $\mathbf{X}^c_j$ denotes all the system's variables except $X_j$. The difference between the joint entropy and the sum of all exclusive information  terms, $H_{(1)}$,  defines a quantity known~\cite{te1978nonnegative} as the \textit{dual total correlation} (DTC)\footnote{The DTC is also known as \textit{excess entropy} in \cite{olbrich2008should}, whose definition differs from its typical use in the context of time series, e.g. \cite{Crutchfield2003}.}:
\begin{equation}\label{eq:DTCC}
\text{DTC} = H(\mathbf{X}) - H_{(1)},
\end{equation}
which measures the portion of the joint entropy that is shared between two or more variables of the system. When $N=2$ then  $\text{DTC}=I(X_1;X_2)$, and hence the $\text{DTC}$ has also been suggested in the literature as a measure for the multivariate mutual information.

By comparing \eqref{eq:amari} and \eqref{eq:DTCC}, it would be appealing to look for a decomposition of the DTC of the form $\text{DTC}= \sum_{k=2}^N \Delta H_{(k)}$, where $\Delta H_{(k)}\geq 0$ would measure the information that is shared by exactly $k$ variables~\cite{Rosas2015benelux}. With this, one could define an \emph{internal} entropy $H_{(j)} = H_{(1)} + \sum_{i=2}^j \Delta H_{(i)}$ as the information that is shared between at most  $j$ variables, in contrast to the \emph{external} entropy $H^{(j)} = H^{(1)} - \sum_{i=2}^j \Delta H^{(i)}$ which describes the information provided by the $j$-marginals. These entropies form a non-decreasing sequence:
\begin{equation}
H_{(1)} \leq \dots \leq H_{(N-1)} \leq H(\mathbf{X}) \leq H^{(N-1)} \leq \dots \leq H^{(1)}
\enspace.
\end{equation}
This layered structure, and its relationship with the $\text{TC}$ and the $\text{DTC}$, is graphically represented in Figure~\ref{fig:spectrum}.
\begin{figure}[ht]

                \centering
                \includegraphics[width=1\columnwidth]{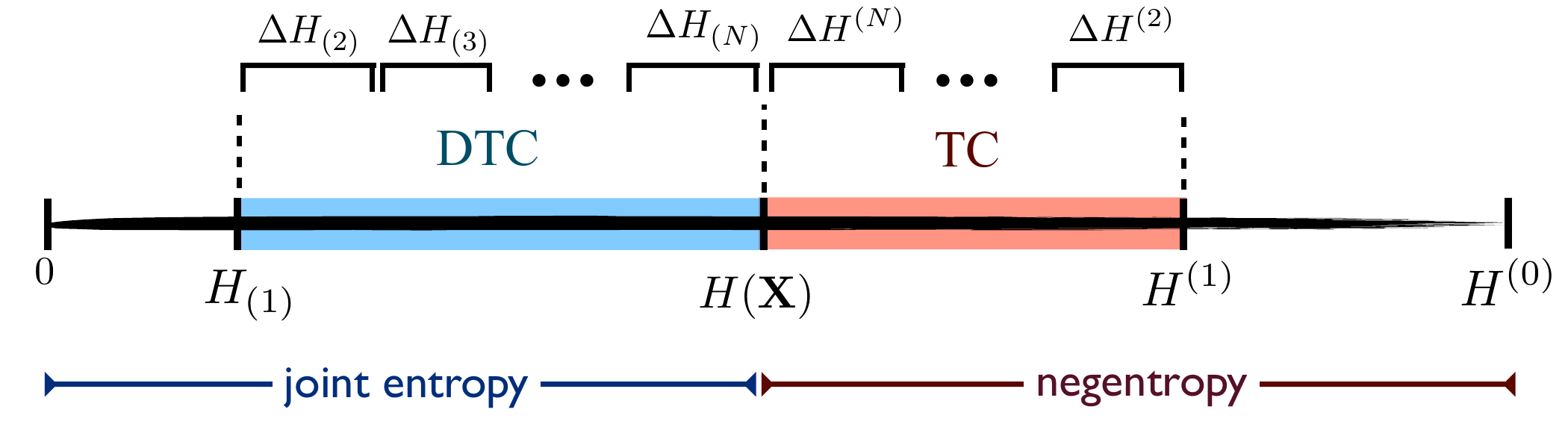}
                \caption{Layers of internal and external entropies that decompose the DTC and the TC. Each $\Delta H^{(j)}$ shows how much information is contained in the $j$-marginals, while each $\Delta H_{(j)}$ measures the information is shared between exactly $j$ variables.}
                \label{fig:spectrum}

\end{figure}

It is interesting to note that even though the TC and  DTC coincide for the case of $N=2$, these quantities are in general different for larger system sizes. Therefore, in general $\Delta H_{(k)} \neq \Delta H^{(k)}$, although it is appealing to believe that there should exist a relationship between them. One of the goals of this paper is to explore the difference between these quantities.

\subsection{Inclusion-exclusion decompositions}
\label{sec:2c}

Perhaps the most natural approach to decompose the DTC and joint entropy is to apply the inclusion-exclusion principle, using a simplifying analogy that the entropies and areas have similar properties. A refined version of this approach can be found in and also in the \textit{I-measures}~\cite{yeung1991} and in the \textit{multi-scale complexity}~\cite{bar2004multiscale}. For the case of three variables, this approach gives
\begin{equation}\label{eq:asdasd}
\text{DTC}_{N=3} = I(X_1;X_2|X_3) + I(X_2;X_3|X_1) + I(X_3;X_1|X_2) + I(X_1;X_2;X_3)~.
\end{equation}
The last term is known as the \textit{co-information}~\cite{Bell2003} (being closely related to the \textit{interaction information}~\cite{mcgill1954multivariate}), and can be defined using the inclusion-exclusion principle as
\begin{align}
 I(X_1;X_2;X_3) \coloneqq &H(X_1) + H(X_2) + H(X_3) - H(X_1,X_2) - H(X_2,X_3) \nonumber\\ &- H(X_1,X_3) + H(X_1,X_2,X_3) \\
= &I(X_1;X_2)- I(X_1;X_2|X_3)
\enspace.
\end{align}
As $I(X_1;X_2;X_2) = I(X_1;X_2)$, the co-information has also been proposed as a candidate for extending the mutual information to multiple variables. For a summary of the various possible extensions of the mutual information, see Table~\ref{table1} and also additional discussion in
Ref.~\cite{James2011}.
\begin{table}[h]
\caption{Summary of the candidates for extending the mutual information for $N\geq 3$.}
\label{table1}
\begin{center}
\begin{tabular}{|c|c|} \hline \hline
Name                           & Formula \\ \hline\hline
\textit{Total correlation}           & $\text{TC}   = \sum_j H(X_j) - H(\mathbf{X}) $    \\
\textit{Dual total correlation}    & $\text{DTC} = H(\mathbf{X}) - \sum_j H(X_j|\mathbf{X}_j^c) $    \\
\textit{Co-information}            & $I(X_1;X_2;X_3)  = I(X_1;X_2) - I(X_1;X_2|X_3) $    \\
\hline\hline
\end{tabular}
\end{center}
\end{table}

It is tempting to coarsen the decomposition provided by this approach in order to build a decomposition for the DTC. In this decomposition, the co-information associates to $\Delta H_{(3)}$, and the the remaining terms of \eqref{eq:asdasd} associate to $\Delta H_{(2)}$. With this, one can build a Venn diagram for the information sharing between three variables, as in Figure~\ref{fig:Venn}. However, the resulting decomposition and diagram are not very intuitive since the co-information can be negative.
\begin{figure}[ht]
    \centering
    \includegraphics[width=.8\columnwidth]{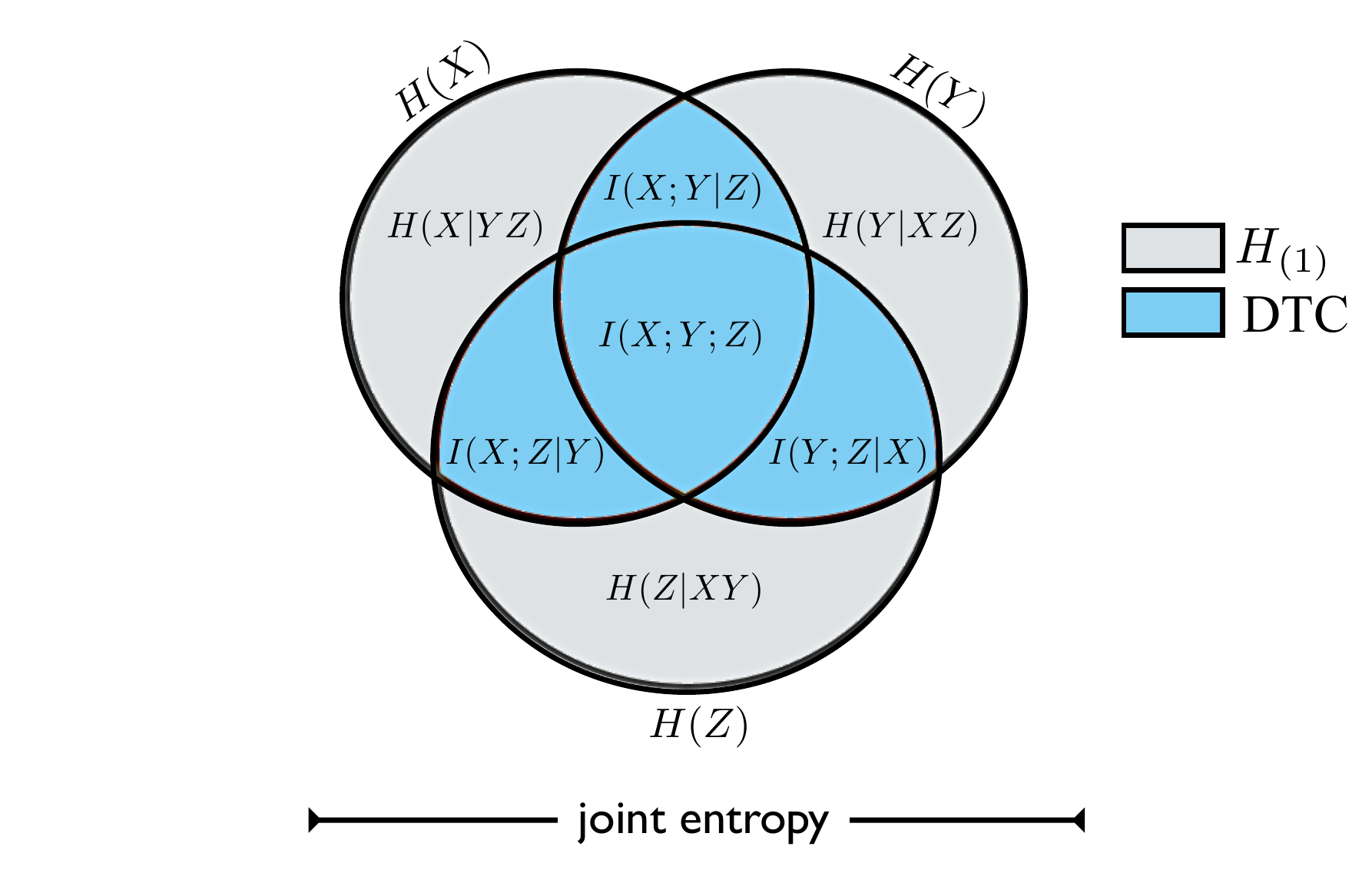}
    \caption{An approach based on the \textit{I-measures} decomposes the total entropy of three variables $H(X,Y,Z)$ into $7$ signed areas.}
    \label{fig:Venn}
\end{figure}

As part of this temptation, it is appealing to consider the conditional mutual information $I(X_1;X_2|X_3)$ as the information contained in $X_1$ and $X_2$ that is not contained in $X_3$, just as the conditional entropy $H(X_1|X_2)$ is the information that is in $X_1$ and not in $X_2$. However, the latter interpretation works because conditioning always reduces entropy (i.e., $H(X_1)\geq H(X_1|X_2)$) while this is not true for mutual information; that is, in some cases the conditional mutual information $I(X_1;X_2|X_3)$ can be greater than $I(X_1;X_2)$. This suggests that the conditional mutual information can capture information that extends beyond $X_1$ and $X_2$, incorporating higher-order effects with respect to $X_3$. Therefore, a better understanding of the conditional mutual information is required in order to refine the decomposition suggested by~\eqref{eq:asdasd}.

\subsection{Synergistic information}
\label{sec:2d}

An extended treatment of the conditional mutual information and its relationship with the mutual information decomposition can be found in \cite{Griffith2014a,griffith2014thesis}. For presenting these ideas, let us consider two random variables $X_1$ and $X_2$ which are used to predict $Y$. The \textit{total predictability}\footnote{Note that the term \emph{total predictability} has also been used in \cite{Crutchfield2003} with a definition that differs from our current usage.}, i.e., the part of the randomness of $Y$ that can be predicted by $X_{1}$ and $X_{2}$, can be expressed using the chain rule of the mutual information as\footnote{For simplicity, through the paper we use the shorthand notation $XY=(X,Y)$.}
 \begin{equation}\label{eq:chain}
I(X_1X_2;Y) = I(X_1;Y) + I(X_2;Y|X_1)
\enspace.
\end{equation}
It is natural to think that the predictability provided by $X_1$, which is given by the term $I(X_1;Y)$, can be either \textit{unique} or \textit{redundant} with respect of the information provided by $X_2$. On the other hand, due to \eqref{eq:chain} is clear that the unique predictability contributed by $X_2$ must be contained in $I(X_2;Y|X_1)$. However, the fact that $I(X_2;Y|X_1)$ can be larger than $I(X_2;Y)$ ---while the latter contains both the unique and redundant contributions of $X_2$--- suggests that there can be an additional predictability  that is accounted for only by the conditional mutual information.

Following this rationale, we denote as \textit{synergistic predictability} the part of the conditional mutual information that corresponds to evidence about the target that is not contained in any single predictor, but is only revealed when both are known. As an example of this, consider again the case in which $X_1$ and $X_2$ are independent random bits and $Y = X_1 \oplus X_2$. Then, it can be seen that $I(X_1;Y) = I(X_{2};Y) = 0$ but $I(X_{1}X_{2}; Y ) = I(X_{1};Y|X_{2}) = 1$. Hence, neither $X_1$ nor $X_2$ individually provide information about $Y$, although together they fully determine it.

Further discussions about the notion of information synergy can be found in \cite{williams2010,griffith2014b,bertschinger2014,olbrich2015information}.


\section{A non-negative joint entropy decomposition}
\label{sec:3}

Following the discussion presented in Section~\ref{sec:2b}, we search for a decomposition of the joint entropy that reflects the private, common and synergistic modes of information sharing. In this way, we want the decomposition to distinguish information that is shared only by few variables from information that accessible from the entire system.

Our framework is based on distinguishing the directed notion of \textit{predictability} from the undirected one of \textit{information}. It is to be noted that there is an ongoing debate about the best way of characterizing and computing the predictability in arbitrary systems, as the commonly used axioms are not enough for specifying a unique formula that satisfies them~\cite{griffith2014b}. Nevertheless, our approach is to explore how far one can reach based an axiomatic approach. In this way, our results are going to be consistent with any choice of formula that is consistent with the discussed axioms.

In the following, Sections~\ref{sec:axioms}, \ref{sec:infoo} and \ref{sec:IIIC} discuss the basic features of predictability and information. After these necessary preliminaries, Section~\ref{sec':dec} finally presents our joint entropy decomposition for discrete and continuous variables.

\subsection{Predictability axioms}
\label{sec:axioms}

Let us consider two variables $X_1$ and $X_2$ that are used to predict a target variable $Y \coloneqq X_3$. Intuitively, $I(X_1;Y)$ quantifies the predictability of $Y$ that is provided by $X_1$. In the following, we want to find a function $\dr{X_1X_2}{Y}$ that measures the \textit{redundant predicability} provided by $X_1$ with respect to the predictability provided by $X_2$, and a function $\du{X_1}{Y|X_2}$ that measures the \textit{unique predictability} that is provided by $X_1$ but not by $X_2$. Following~\cite{Griffith2014a}, we first determine a number of desired properties that these functions should have.

\begin{definition}
A \textit{predictability decomposition} is defined by the real-valued functions $\dr{X_1X_2}{Y}$ and $\du{X_1}{Y|X_2}$ over the distributions of $(X_1,Y)$ and $(X_2,Y)$, which satisfy the following axioms:
\begin{itemize}

\item[(1)] Non-negativity: $\dr{X_1X_2}{Y},\: \du{X_1}{Y|X_2} \geq 0$.

\item[(2)] $I(X_1;Y) = \dr{X_1X_2}{Y} + \du{X_1}{Y|X_2}$.

\item[(3)] $I(X_1X_2;Y) \geq  \dr{X_1X_2}{Y} + \du{X_1}{Y|X_2} + \du{X_2}{Y|X_1}$.

\item[(4)] \textit{Weak symmetry I}: $\dr{X_1 X_2}{Y}=\dr{X_2 X_1}{Y}$.

\end{itemize}
\end{definition}

Above, Axiom (3) states that the sum of the redundant and corresponding unique predictabilities given by each variable cannot be larger than the total predictability\footnote{In fact, the difference between the right and left hand terms of Axiom (3) gives the \textit{synergistic predictability}, whose analysis will not be included in this work.}. Axiom (4) states that the redundancy is independent of the ordering of the predictors. The following Lemma determines the bounds for the redundant predicability (the proof is given in Appendix~\ref{ap:lemma1}).
\begin{lemma}\label{lemma123}
The functions $\dr{X_1X_2}{Y}$ and $\du{X_1}{Y|X_2}=I(X_1;Y) - \dr{X_1X_2}{Y}$ satisfy Axioms (1)--(3) if and only if
\begin{equation}
\min\{ I(X_1;Y), I(X_2;Y) \} \geq \dr{X_1X_2}{Y} \geq [ I(X_1;X_2;Y) ]^+
\enspace,
\end{equation}
where $[a]^+ = \max\{ a ,0\}$.
\end{lemma}
\begin{corollary}\label{coro2}
There always exists at least one predictability decomposition that satisfies Axioms (1)--(4), which is given by
\begin{equation}\label{eq:sasdasqwe}
\dr{X_1X_2}{Y} \coloneqq \min\{ I(X_1;Y), I(X_2;Y) \}.
\end{equation}
\end{corollary}
\begin{proof}
Being a symmetric function on $X_1$ and $X_2$, \eqref{eq:sasdasqwe} satisfies Axiom (4). Also, as \eqref{eq:sasdasqwe} is equal to the upper bound given in Lemma~\ref{lemma123}, Axioms (1)--(3) are satisfied due to Lemma~\ref{lemma1}.
\end{proof}

In principle, the notion of redundant predictability takes the point of view of the target variable and measures the parts that can be predicted by both $X_1$ and $X_2$ when they are used by themselves, i.e., without combining them with each other. It is appealing to think that there should exist a unique function that provides such a measure. Nevertheless, these axioms define only very basic properties that a measure of redundant predictability should satisfy, and hence in general they are not enough for defining an unique function. In fact, a number of different predictability decompositions have been proposed in the literature \cite{griffith2014b,harder2013,bertschinger2014,barrett2015}.

It is to be noted that, from all the candidates that are compatible with the Axioms, the decomposition given in Corollary~\ref{coro2} gives the largest possible redundant predictability measure. It is clear that in some cases this measure gives an over-estimate of the redundant predictability given by $X_1$ and $X_2$; for an example of this consider $X_1$ and $X_2$ to be independent variables and $Y=(X_1,X_2)$. Nevertheless, \eqref{eq:sasdasqwe} has been proposed as a adequate measure for the redundant predictability of multivariate Gaussians~\cite{barrett2015} (for a corresponding discussion see Section~\ref{sec:gaussian}).

\subsection{Shared, private and synergistic information}
\label{sec:infoo}

Let us now introduce an additional axiom, which will form the basis for our proposed \textit{information decomposition}.

\begin{definition}
A \textit{symmetrical information decomposition} is given by the real valued functions $\Red{X_1;X_2;X_3}$ and $\Un{X_1;X_2|X_3}$ over the marginal distributions of $(X_1,X_2)$, $(X_1,X_3)$ and $(X_2,X_3)$, which satisfy Axioms (1) -- (4) for $\Red{X_1;X_2;X_3} \coloneqq \dr{X_1X_2}{X_3}$ and $\Un{X_1;X_2|X_3} \coloneqq \du{X_1}{X_2|X_3}$, while also satisfying the following property:
\begin{itemize}

\item[(5)] \textit{Weak symmetry II}: $\Un{X_1;X_2| X_3}=\Un{X_2;X_1| X_3}$.

\end{itemize}
Finally, $\Syn{X_1;X_2;X_3}$ is defined as $\Syn{X_1;X_2;X_3} \coloneqq I(X_1;X_2|X_3) - \Un{X_1;X_2|X_3}$.

\end{definition}

The role of Axiom (5) can be related to the role of the fifth of Euclid's postulates, as  ---while seeming innocuous--- their addition has strong consequences in the corresponding theory. The following Lemma explains why this decomposition is denoted as symmetrical, and also shows fundamental bounds for these information functions (the proof is presented in Appendix~\ref{app:2}).
\begin{lemma}\label{lemma1}
The functions that compose a symmetrical information decomposition satisfy the following properties:
\begin{itemize}

\item[(a)] \textit{Strong symmetry:}$\;\Red{X_1;X_2;X_3}$ and $\Syn{X_1;X_2;X_3}$ are symmetric on their three arguments.

\item[(b)] \textit{Bounds:} these quantities satisfy the following inequalities:
\begin{align}
&\min\{I(X_1;X_2), I(X_2;X_3), I(X_3;X_1) \} \geq I_\cap(X_1;X_2;X_3) \geq [I(X_1;X_2;X_3)]^+ \label{eq:bounds}\\
&\min\{I(X_1;X_3), I(X_1;X_3|X_2) \} \geq \Un{X_1;X_3|X_2} \geq 0\\
&\min\{I(X_1;X_2|X_3), I(X_2;X_3|X_1), I(X_3;X_1|X_2) \} \geq I_\text{S}(X_1;X_2;X_3) \geq [-I(X_1;X_2;X_3) ]^+
\end{align}
\end{itemize}
\end{lemma}

Note that the defined functions can be used to decompose the following mutual information:
\begin{align}
I(X_1X_2;X_3) &= I(X_1;X_3) + I(X_2;X_3|X_1) \\
I(X_1;X_3) &= \Red{X_1;X_2;X_3} + \Un{X_1;X_3|X_2} \label{eq:g1}\\
I(X_2;X_3|X_1) &= \Un{X_2;X_3|X_1} + \Syn{X_1;X_2;X_3} \label{eq:g2}
\end{align}
In contrast to a decomposition based on the predictability, these measures address properties of the system $(X_1,X_2,X_3)$ as a whole, without being dependent on how it is divided between target  and predictor variables (for a parallelism with respect to the corresponding predictability measures, see Table~\ref{table2}). Intuitively, $\Red{X_1; X_2;X_3}$ measures the \textit{shared information} that is common to $X_1$, $X_2$ and $X_3$; $\Un{X_1;X_3|X_2}$ quantifies the \textit{private information} that is shared by $X_1$ and $X_3$ but not $X_2$, and $\Syn{X_1; X_2;X_3}$ captures the \textit{synergistic information} that exist between $(X_1,X_2,X_3)$. The latter is a non-intuitive mode of information sharing, whose nature we hope to clarify through the analysis of particular cases presented in Sections \ref{sec:indep} and \ref{sec:gaussian}.
\begin{table}[h]
\caption{Parallelism between predictability and information measures.}
\label{table2}
\begin{center}
\begin{tabular}{|c|c|} \hline \hline
Directed measures                                                           &    Symmetrical measures \\ \hline\hline
\textit{Redundant predictability} $\dr{X_1X_2}{X_3}$            &  \textit{Shared information} $\Red{X_1;X_2;X_3}$    \\
\textit{Unique predictability}  $\du{X_1}{X_2|X_3}$               &  \textit{Private information} $\Un{X_1;X_2|X_3}$    \\
\textit{Synergistic predictability}                                        &  \textit{Synergistic information}  $\Syn{X_1;X_2;X_3}$  \\
\hline\hline
\end{tabular}
\end{center}
\end{table}

Note also that the co-information can be expressed as
\begin{equation}
\label{eq:g3}
I(X_1;X_2;X_3) = \Red{X_1;X_2;X_3} - \Syn{X_1;X_2;X_3}
\enspace.
\end{equation}
Hence, a strictly positive (resp. negative) co-information is a sufficient ---although not necessary--- condition for the system to have a non-zero shared (resp. synergistic) information.

\subsection{Further properties of the symmetrical decomposition}
\label{sec:IIIC}

At this point, it is important to clarify a fundamental distinction that we make between the notions of \textit{predictability} and \textit{information}. The predictability is intrinsically a directed notion, which is based on a distinction between predictors and the target variable. On the contrary, we use the term information to exclusively refer to intrinsic statistical properties of the whole system which do not rely on such distinction. The main difference between the two notions is that, in principle, the predictability only considers the predictable parts of the target, while the shared information also considers the joint statistics of the predictors. Although this distinction will be further developed when we address the case of Gaussian variables (c.f. Section~\ref{sec:gaussians4}), let us for now present a simple example to help developing intuitions about this issue.
\begin{example}
Define the following functions:
\begin{align}
\Red{X_1;X_2;X_3} &= \min\{ I(X_1;X_2), I(X_2;X_3), I(X_3;X_1) \} \\
\Un{X_1;X_2|X_3} &= I(X_1;X_2) - \Red{X_1;X_2;X_3}
\end{align}
It is straightforward that these functions satisfy Axioms (1)--(5), and therefore constitute a symmetric information decomposition. In contrast to the decomposition given in Corollary~\ref{coro2}, this can be seen to be strongly symmetric and also dependent on the three marginals $(X_1,X_2)$, $(X_2,X_3)$ and $(X_1,X_3)$.
\end{example}

In the following Lemma we will generalize the previous construction, whose simple proof is omitted.
\begin{lemma}\label{lemaaaa}
For a given predictability decomposition with functions $\dr{X_1X_2}{X_3}$ and $\du{X_1}{X_2|X_3}$, the functions
\begin{align}
\Red{X_1;X_2;X_3} &= \min\{ \dr{X_1X_2}{X_3}, \dr{X_2X_3}{X_1}, \dr{X_3X_1}{X_2} \} \\
\Un{X_1;X_2|X_3} &= I(X_1;X_2) - \Red{X_1;X_2;X_3}
\end{align}
provide a symmetrical information decomposition, which is called the \textit{canonical symmetrization of the predictability}.
\end{lemma}
\begin{corollary}
There always exists at least one symmetric information decomposition.
\end{corollary}
\begin{proof}
This is a direct consequence of the previous Lemma and Corollary~\ref{coro2}.
\end{proof}

Maybe the most remarkable property of symmetrized information decompositions is that, in contrast to directed ones, they are uniquely determined by Axioms (1)--(5) for a number of interesting cases.
\begin{theorem}\label{theo}
The symmetric information decomposition is unique if the variables form a Markov chain or two of them are pairwise independent.
\end{theorem}
\begin{proof}
Let us consider the upper and lower bound for $I_\cap$ given in \eqref{eq:bounds}, denoting them as $c_{1}:= [I(X_1;X_2;X_3)]^+$ and $c_{2}:=  \min\{I(X_1;X_2), I(X_2;X_3), I(X_1;X_3) \}$. These bounds restrict the possible $I_\cap$ functions to lay in the interval $[c_{1},c_{2}]$ of length
\begin{align}
|c_{2} - c_{1}| = \min\{ &I(X_1;X_2), I(X_2;X_3), I(X_1;X_3),\\
&I(X_1;X_2|X_3), I(X_2;X_3|X_1), I(X_3;X_1|X_2) \}
\enspace.
\end{align}
Therefore, the framework will provide a unique expression for the shared information if (at least) one of the above six terms is zero. These scenarios correspond either to Markov chains, where one conditional mutual information term is zero, or pairwise independent variables where one mutual information term vanishes.
\end{proof}

Pairwise independent variables and Markov chains are analyzed in Sections~\ref{sec:indep} and \ref{sec:PME}, respectively.

\subsection{Decomposition for the joint entropy of three variables}\label{sec':dec}

Now we use the notions of redundant, private and synergistic information functions for developing a non-negative decomposition of the joint entropy, which is based on a non-negative decomposition of the DTC. For the case of three discrete variables, by applying \eqref{eq:g2} and \eqref{eq:g3} to \eqref{eq:asdasd}, one finds that
\begin{align}
\text{DTC} =\;& \Un{X_1;X_2|X_3} + \Un{X_2;X_3|X_1} + \Un{X_3;X_1|X_2} \nonumber\\
&+ \Red{X_1;X_2;X_3} + 2 \Syn{X_1;X_2;X_3} \label{eq:asdapqwoe}
\enspace.
\end{align}
From \eqref{eq:DTCC} and \eqref{eq:asdapqwoe}, one can propose the following decomposition for the joint entropy:
\begin{equation}\label{eq:asdadasdasdad}
H(X_1,X_2,X_3) = H_{(1)} +  \Delta H_{(2)} +\Delta H_{(3)}.
\end{equation}
where
\begin{align}
H_{(1)} &=   H(X_1|X_{2},X_{3}) + H(X_2|X_{1},X_{3}) +H(X_3|X_{1},X_{2}) \\
\Delta H_{(2)} &= \Un{X_1;X_2|X_3} + \Un{X_2;X_3|X_1} + \Un{X_3;X_1|X_2} \label{eq:2}\\
\Delta H_{(3)} &= \Red{X_1;X_2;X_3} + 2 \Syn{X_1;X_2;X_3} \label{eq:3}
\end{align}
In contrast to \eqref{eq:asdasd}, here each term is non-negative because of Lemma~\ref{lemma1}\footnote{From \eqref{eq:g2}, it can be seen that the co-information is sometimes negative for compensating the triple counting of the synergy due to the sum of the three conditional mutual information terms.}. Therefore, \eqref{eq:asdadasdasdad} yields a non-negative decomposition of the joint entropy, where each of the corresponding terms captures the information that is shared by one, two or three variables. Interestingly, $H_{(1)}$ and $\Delta H_{(2)}$ are homogeneous (being the sum of all the exclusive information or private information of the system) while $\Delta H_{(3)}$ is composed by a mixture of two different information sharing modes.

An analogous decomposition can be developed for the case of continuous random variables. Nevertheless, as the differential entropy can be negative, not all the terms of the decomposition can be non-negative. In effect, following the same rationale that lead to \eqref{eq:asdadasdasdad}, the following decomposition can be found:
\begin{equation}
h(X_1,X_2,X_3) = h_{(1)} + \Delta H_{(2)} + \Delta H_{(3)}.
\end{equation}%
Above, $h(X)$ denotes the differential entropy of $X$, $\Delta H_{(2)}$ and $\Delta H_{(3)}$ are as defined in \eqref{eq:2} and \eqref{eq:3}, and
\begin{equation}
h_{(1)} = h(X_1|X_2X_3) + h(X_2|X_1X_3) +h(X_3|X_1X_2)\enspace.
\end{equation}
Hence, although both the joint entropy $h(X_1,X_2,X_3)$ and $h_{(1)}$ can be negative, the remaining terms conserve their non-negative condition.

It can be seen that the lowest layer of the decomposition is always trivial to compute, and hence the challenge is to find expressions for $\Delta H_{(2)}$ and $\Delta H_{(3)}$. In the rest of the paper, we will explore scenarios were these quantities can be characterized.


\section{Pairwise independent variables}
\label{sec:indep}

In this section we focus on the case where two variables are pairwise independent while being globally connected by a third variable. The fact that pairwise independent variables can become correlated when additional information becomes available is known in statistics literature as the \textit{Bergson's paradox} or \textit{selection bias}~\cite{berkson1946}, or as the \textit{explaining away effect} in the context of artificial intelligence~\cite{kim1983}. As an example of this phenomenon, consider $X_1$ and $X_2$ to be two pairwise independent canonical Gaussians variables, and $X_3$ a binary variable that is equal to $1$ if $X_1+X_2 > 0$ and zero otherwise. Then, knowing that $X_3=1$ implies that $X_2 > -X_1$, and hence knowing the value of $X_1$ effectively reduces the uncertainty about $X_2$.

In our framework, Bergson's paradox can be understood as synergistic information that is introduced by the third component of the system. In fact, we will show that in this case the synergistic information function is unique and given by
\begin{equation}
\Syn{X_1;X_2;X_3} = \sum_{x_3} p_{X_3}(x_3) I(X_1;X_2|X_3=x_3) = I(X_1;X_2|X_3)\, ,
\end{equation}
which is, in fact, a measure of the dependencies between $X_1$ and $X_2$ that are created by $X_3$. In the following, Section~\ref{sec:iva} presents the unique symmetrized information decomposition for this case. Then, Section~\ref{sec:4b} focuses on the particular case where $X_3$ is a function of the other two variables.

\subsection{Uniqueness of the entropy decomposition}
\label{sec:iva}

Let us assume that $X_1$ and $X_2$ are pairwise independent, and hence the joint p.d.f. of $X_1$, $X_2$ and $X_3$ has the following structure:
\begin{equation}\label{eq:pipip}
p_{X_1X_2X_3}(x_1,x_2,x_3) = p_{X_1}(x_1)p_{X_2}(x_2) p_{X_3|X_1X_2}(x_3|x_1,x_2)
\enspace.
\end{equation}
It is direct to see that in this case $p_{X_1X_2}=\sum_{x_3}p_{X_1X_2X_3}=p_{X_1}p_{X_2}$, but $p_{X_1X_2|X_3}\neq p_{X_1|X_3}p_{X_2|X_3}$. Therefore, as $I(X_1;X_2)=0$, it is direct from Axiom (1) that any redundant predictability function satisfies $\dr{X_1X_3}{X_2}=\dr{X_2X_3}{X_1}=0$. However, the axioms are not enough to uniquely determine $\dr{X_1X_2}{X_3}$\footnote{Note that in this case $I(X_1;X_2;X_3) = - I(X_1;X_2|X_3)\leq 0$, the only restriction that the bound presented in Lemma~\ref{lemma1} provides is $\min\{ I(X_1;X_3), I(X_2;X_3) \} \geq \dr{X_1X_2}{X_3} \geq 0$.}. Nevertheless, the symmetrized decomposition is uniquely determined, as shown in the next Corollary that is a consequence of Theorem~\ref{theo}.
\begin{corollary}\label{lemmaPIP}
If $X_1$, $X_2$ and $X_3$ follow a p.d.f. as \eqref{eq:pipip}, then the shared, private and synergetic information functions are unique. They are given by
\begin{align}
\Red{X_1;X_2;X_3} &= \Un{X_1;X_2|X_3}  = 0 \label{eqlema1}\\
\Un{X_1;X_3|X_2} &= I(X_1;X_3) \label{eqlema2}\\
\Un{X_2;X_3|X_1} &= I(X_2;X_3) \label{eqlema3}\\
\Syn{X_1;X_2;X_3} &= I(X_1;X_2|X_3) = - I(X_1;X_2;X_3). \label{eqlema4}
\end{align}
\end{corollary}
\begin{proof}
The fact that there is no shared information follows directly from the upper bound presented in Lemma~\ref{lemma1}. Using this, the expressions for the private information can be found using Axiom (2). Finally, the synergistic information can be computed as
\begin{equation}
\Syn{X_1;X_2;X_3} = I(X_1;X_2|X_3) - \Un{X_1;X_2|X_3} = I(X_1;X_2|X_3)
\enspace.
\end{equation}
The second formula for the synergistic information can be found then using the fact that $I(X_1;X_2)=0$.
\end{proof}

With this corollary, the unique decomposition of the $\text{DTC}=\Delta H_{(2)} + \Delta H_{(3)}$ can be found to be
\begin{align}
\Delta H_{(2)}&= I(X_1;X_3) + I(X_2;X_3) \\
\Delta H_{(3)}&= 2 I(X_1;X_2|X_3) \label{eq:synn}
\enspace.
\end{align}
Note that the terms $\Delta H_{(2)}$ and $\Delta H_{(3)}$ can be bounded as follows:
\begin{align}
\Delta H_{(2)}& \leq \min\{ H(X_1), H(X_3) \} + \min\{ H(X_2), H(X_3) \} \label{eq:bound00}
\enspace,\\
\Delta H_{(3)}& \leq 2 \min\{ H(X_1|X_3), H(X_2|X_3) \} \label{eq:bound01}
\enspace.
\end{align}
The bound for $\Delta H_{(2)}$ follows from the basic fact that $I(X;Y) \leq \min\{ H(X), H(Y) \}$. The second bound follows from
\begin{align}
I(X;Y|Z) &= \sum_z p_Z(z) I(X;Y|Z=z) \\
&\leq \sum_z p_Z(z) \min\left\{ H(X|Z=z), H(Y|Z=z) \right\} \\
&\leq \min\left\{ \sum_z p_Z(z) H(X|Z=z), \sum_z p_Z(z) H(Y|Z=z) \right\} \\
&= \min\{ H(X|Z), H(Y|Z) \} \label{eq:0123124}
\enspace.
\end{align}

\subsection{Functions of independent arguments}
\label{sec:4b}

Let us focus in this section on the special case where $X_3 = F(X_1,X_2)$ is a function of two independent random inputs, and study its corresponding entropy decomposition. We will consider $X_1$ and $X_2$ as inputs and $F(X_1,X_2)$ to the output. Although this scenario fits nicely in the predictability framework, it can also be studied from the shared information framework's perspective. Our goal is to understand how $F$ affects the information sharing structure.

As $H(X_3|X_1,X_2)=0$, we have
\begin{equation}
H_{(1)} = H(X_1|X_2X_3) + H(X_2|X_1X_3)
\enspace.
\end{equation}
The term $H_{(1)}$ hence measures the information of the inputs that is not reflected by the output. An extreme case is given by a constant function $F(X_1,X_2)=k$, for which $\Delta H_{(2)}=\Delta H_{(3)}=0$.

The term $\Delta H_{(2)}$ measures how much of $F$ can be predicted with knowledge that comes from one of the inputs but not from the other. If $\Delta H_{(2)}$ is large then $F$ is not ``mixing'' the inputs too much, in the sense that each of them is by itself able to provide relevant information that is not given also by the other. In fact, a maximal value of $\Delta H_{(2)}$ is given by $F(X_1,X_2)=(X_1,X_2)$, where $H_{(1)} = \Delta H_{(3)} = 0$ and the bound provided in \eqref{eq:bound00} is attained.

Finally, due to \eqref{eq:synn}, there is no shared information and hence $\Delta H_{(3)}$ is just proportional to the synergy of the system. By considering \eqref{eq:bound01}, one finds that $F$ needs to leave some ambiguity about the exact values of the inputs in order for the system to possess synergy. For example, consider a 1-1 function $F$ for which for every output $F(X_1,X_2)=x_3$ one can find the unique values $x_1$ and $x_2$ that generate it. Under this condition $H(X_1|X_3)=H(X_2|X_3)=0$  and hence, because of \eqref{eq:bound01}, is clear that a 1-1 function does not induce synergy. On the other extreme, we showed already that constant functions have $\Delta H_{(3)}=0$, and hence the case where the output of the system gives no information about the inputs also leads to no synergy. Therefore, synergistic functions are those whose output values generate a balanced ambiguity about the generating inputs. To develop this idea further, the next lemma studies the functions that generate a maximum amount of synergy by generating for each output value different 1-1 mappings between their arguments.
\begin{lemma}
Let us assume that both $X_1$ and $X_2$ take values over $\mathcal{K}=\{0,\dots,K-1\}$ and are independent. Then, the maximal possible amount of information synergy is created by the function
\begin{equation}\label{eq:function}
F^*(n,m) = n + m \:(\text{mod}\; K)
\end{equation}%
when both inputs variables are uniformly distributed.
\end{lemma}

\begin{proof}
Using the same rationale than in \eqref{eq:0123124}, it can be shown that if $F$ is an arbitrary function then
\begin{align}
\Syn{X_1;X_2;F(X_1,X_2)} &= I(X_1;X_2|F) \\
&\leq \min\{ H(X_1|F), H(X_2|F) \} \\
&\leq \min\{ H(X_1), H(X_2)\} \\
&\leq \log K \label{eq:boundsss}
\enspace.
\end{align}
where the last inequality follows from the fact that both inputs are restricted to alphabets of size~$K$.

Now, consider $F^*$ to be the function given in \eqref{eq:function} and assume that $X_1$ and $X_2$ are uniformly distributed. It can be seen that for each $z\in\mathcal{K}$ there exist exactly $K$ ordered pairs of inputs $(x_1,x_2)$ such that $F^*(x_1,x_2)=z$, which define a bijection from $\mathcal{K}$ to $\mathcal{K}$. Therefore,
\begin{equation}
 I(X_1;X_2|F=z) = H(X_1|z) - H(X_2|X_1,z) = H(X_1) = \log K\,
\end{equation}
and hence
\begin{equation}
\Syn{X_1;X_2;F^*} = I(X_1;X_2|F^*) = \sum_z \mathbb{P}\{F^{*}=z\}\cdot I(X_1;X_2|F^*=z) = \log K
\enspace,
\end{equation}
showing that the upper bound presented in \eqref{eq:boundsss} is attained.
\end{proof}
\begin{corollary}
The \texttt{XOR} logic gate generates the largest amount of synergistic information possible for the case of binary inputs.
\end{corollary}

The synergistic nature of the addition over finite fields helps to explain the central role it has in various fields. In cryptography, the \textit{one-time-pad} \cite{bloch2011} is an encryption technique that uses finite-field additions for creating a synergistic interdependency between a private message, a public signal and a secret key. This interdependency is completely destroyed when the key is not known, ensuring no information leakage to unintended receivers \cite{shannon1949communication}. Also, in {\em network coding} \cite{ahlswede2000network,li2003linear}, nodes in the network use linear combinations of their received data packets to create and transmit synergistic combinations of the corresponding information messages. This technique has been shown to achieve the multicast capacity in wired communication networks \cite{li2003linear} and has also been used to increase the throughput of wireless systems \cite{katti2008xors}.


\section{Discrete pairwise maximum entropy distributions and Markov chains}
\label{sec:V}

This section studies the case where the system's variables follow a \textit{pairwise maximum entropy} (PME) distribution. These distributions are of great importance in statistical physics and machine learning communities, where they are studied under the names of \textit{Gibbs distributions} \cite{landau1970statistical} or \textit{Markov random fields} \cite{wainwright2008}. 

Concretely, let us consider three pairwise marginal distributions $p_{X_1X_2}, p_{X_2X_3}$ and $p_{X_1X_3}$ for the discrete variables $X_1$, $X_2$ and $X_3$. Let us denote as $\mathcal{Q}$ the set of all the joint p.d.f.s over $(X_1,X_2,X_3)$ that have those as their pairwise marginals distributions. Then, the corresponding PME distribution is given by the joint p.d.f. $\tilde{p}_{\mathbf{X}}(x_1,x_2,x_3)$ that satisfies
\begin{equation} \label{maxx}
\tilde{p}_\mathbf{X} = \argmax_{p \in \mathcal{Q}} H(\{ p \})
\enspace.
\end{equation}
For the case of binary variables (i.e. $X_j\in \{0,1\}$), the PME distribution is given by an Ising distribution \cite{cipra1987}:
\begin{equation}\label{eq:ising}
 \tilde{p}_{\mathbf{X}}(\mathbf{X}) = \frac{e^{- \mathcal{E}({\mathbf{X}})}}{ Z}
\enspace,
\end{equation}
where $Z$ is a normalization constant and $\mathcal{E}(\mathbf{X})$ an \textit{energy function} given by $\mathcal{E}(\mathbf{X}) = \sum_{i} J_i X_i + \sum_j\sum_{k\neq j} J_{j,k} X_jX_k$, being $J_{j,k}$ the coupling terms. In effect, if $J_{i,k}=0$ for all $i$ and $k$, then $\tilde{p}_{\mathbf{X}}(\mathbf{X})$ can be factorized as the product of the unary-marginal p.d.f.s.


In the context of the framework discussed in Section~\ref{sec:2.1}, a PME system has $\text{TC}=\Delta H^{(2)}$ while $\Delta H^{(3)}=0$. In contrast, Section~\ref{sec:PME} studies these systems under the light of the decomposition of the $\text{DTC}$ presented in Section~\ref{sec':dec}. Then, Section~\ref{sec:markov} specifies the analysis for the particular case of Markov chains.

\subsection{Synergy minimization}
\label{sec:PME}

It is tempting to associate the synergistic information with that which is only in the joint p.d.f. but not in the pairwise marginals, i.e. with $\Delta H^{(3)}$. However, the following result states that there can exist some synergy defined by the pairwise marginals themselves.
\begin{theorem}\label{eq:PME}
PME distributions have the minimum amount of synergistic information that is allowed by their pairwise marginals.
\end{theorem}
\begin{proof}
Note that
\begin{align}
\max_{p\in \mathcal{Q}} H(X_1X_2X_3) &= H(X_1X_2) + H(X_3) - \min_{p\in \mathcal{Q}} I(X_1X_2;X_3) \\
&= H(X_1X_2) + H(X_3) - I(X_1;X_3) - \min_{p\in \mathcal{Q}} I(X_2;X_3|X_1) \\
&= H(X_1X_2) + H(X_3) - I(X_1;X_3) - \Un{X_2;X_3|X_1} - \min_{p\in \mathcal{Q}} \Syn{X_1;X_2;X_3} \label{eqwsa}
\enspace.
\end{align}
Therefore, maximizing the joint entropy for fixed pairwise marginals is equivalent to minimizing the synergistic information. Note that the last equality follows from the fact that $\Un{X_2;X_3|X_1}$ by definition only depends on the pairwise marginals.
\end{proof}
\begin{corollary}
For an arbitrary system $(X_1,X_2,X_3)$, the synergistic information can be decomposed as
\begin{equation}
\Syn{X_1;X_2;X_3} = I_\text{S} ^\text{PME} + \Delta H^{(3)}
\end{equation}
where $\Delta H^{(3)}$ is as defined in \eqref{eq:amari} and $I_\text{S} ^\text{PME} = \min_{p\in \mathcal{Q}} \Syn{X_1;X_2;X_3}$ is the synergistic information of the corresponding PME distribution.
\end{corollary}
\begin{proof}
This can be proven noting that, for an arbitrary p.d.f. $p_{X_1X_2X_3}$, it can be seen that
\begin{align}
\Delta H^{(3)} =& \max_{p\in \mathcal{Q}} H(X_1X_2X_3) - H(\{ p_{X_1X_2X_3} \}) \\
=& I_\text{S}(\{ p_{X_1X_2X_3} \}) -  \min_{p\in \mathcal{Q}} \Syn{X_1;X_2;X_3}
\enspace.
\end{align}
Above, the first equality corresponds to the definition of $\Delta H^{(3)}$ and the second equality comes from using \eqref{eqwsa} on each joint entropy term and noting that only the synergistic information depends on more than the pairwise marginals.
\end{proof}

The previous corollary shows that $\Delta H^{(3)}$ measures only one part of the information synergy of a system, the part that can be removed without altering the pairwise marginals. Note that PME systems with non-zero synergy are easy to find. For an example, consider $X_1$ and $X_2$ to be two independent equiprobable bits, and $X_3 = X_1\, \texttt{AND} \,X_2$. It can be shown that for this case one has $\Delta H^{(3)} =0$ \cite{schneidman2003network}. On the other side, as the inputs are independent the synergy can be computed using \eqref{eqlema4}, and therefore a direct calculation shows that
\begin{equation}
\Syn{X_1;X_2;X_3} = I(X_1;X_2|X_3) = H(X_1|X_3) - H(X_1|X_2X_3) = 0.1887
\enspace.
\end{equation}

From the previous discussion, one can conclude that only a special class of pairwise distributions $p_{X_1X_2}, p_{X_1X_3}$, and $p_{X_2X_3}$ are compatible with having null synergistic information in the system. This is a remarkable result, as the synergistic information is usually considered to be an effect purely related to high-order marginals. It would be interesting to have an expresion for the minimal information synergy that a set of pairwise distributions requires, or equivalently, a symmetrized information decomposition for PME distributions. A particular case that allows a unique solution is discussed in the next section.

\subsection{Markov chains}
\label{sec:markov}

Markov chains maximize the joint entropy subject to constrains on only two of the three pairwise distributions. In effect, following the same rationale as in the proof of Theorem~\ref{eq:PME}, it can be shown that
\begin{equation}
H(X_1,X_2,X_3) = H(X_1X_2) + H(X_3) - I(X_2;X_3) - I(X_1;X_3|X_2)
\enspace.
\end{equation}
Then, for fixed pairwise distributions $p_{X_1X_2}$ and $p_{X_2X_3}$, maximizing the joint entropy is equivalent to minimizing the conditional mutual information. Moreover, the maximal entropy is attained by the p.d.f. that makes $I(X_1;X_3|X_2)=0$, which is precisely the Markov chain $X_1-X_2-X_3$ with joint distribution
\begin{equation}\label{eq:markov01}
p_{X_1X_2X_3} = \frac{p_{X_1X_2} p_{X_2X_3}}{p_{X_2}}
\enspace.
\end{equation}
For the binary case, it can be shown that a Markov chain corresponds to an Ising distribution like \eqref{eq:ising}, where the interaction terms $J_{1,3}$ is equal to zero.

Theorem~\ref{theo} showed that the symmetric information decomposition for Markov chains is unique. We develop this decomposition in the following corollary.
\begin{corollary}\label{lemmaMarkov}
If $X_1-X_2-X_3$ is a Markov chain, then their unique shared, private and synergistic information functions are given by
\begin{align}
\Red{X_1;X_2;X_3} &= I(X_1;X_3) \label{eqlema01}\\
\Un{X_1;X_2|X_3} &= I(X_1;X_2) - I(X_1;X_3) \label{eqlema02}\\
\Un{X_2;X_3|X_1} &= I(X_2;X_3) - I(X_1;X_3) \label{eqlema03}\\
\Syn{X_1;X_2;X_3} &= \Un{X_1;X_3|X_2}=0. \label{eqlema04}
\end{align}
In particular, Markov chains have no synergistic information.
\end{corollary}
\begin{proof}
For this case one can show that
\begin{equation}\label{eq:ewrwerewR}
\min_{\substack{i,j\in\{1,2,3\}\\i\neq j}}\{ I(X_i;X_j) \} = I(X_1;X_3) = I(X_1;X_2;X_3)
\enspace,
\end{equation}
where the first equality is a consequence of the data process inequality, and the second of the fact that $I(X_1;X_3|X_2)=0$. The above equality shows that the bounds for the shared information presented in Lemma~\ref{lemma1} give the unique solution $\Red{X_1;X_2;X_3}=I(X_1;X_3)$. All the other equalities follow from this fact and their definition.
\end{proof}

Using this corollary, the unique decomposition of the $\text{DTC}=\Delta H_{(2)} + \Delta H_{(3)}$ for Markov chains is given by
\begin{align}
\Delta H_{(2)}&= I(X_1;X_2) + I(X_2;X_3) - 2 I(X_1;X_3)\enspace, \\
\Delta H_{(3)}&=  I(X_1;X_3) \label{eq:synn0}
\enspace.
\end{align}
Hence, corollary~\ref{lemmaMarkov} states that a sufficient condition for three pairwise marginals to be compatible with zero information synergy is for them to satisfy the Markov condition $p_{X_3|X_1} = \sum_{X_2} p_{X_3|X_2} p_{X_2|X_1} $. The question of finding a necessary condition is an open problem, intrinsically linked with the problem of finding a good definition for the shared information for arbitrary PME distributions.

For concluding, let us note an interesting duality that exists between Markov chains and the case where two variables are pairwise independent, which is illustrated in Table~\ref{table3}.
\begin{table}[h]
\caption{Duality between Markov chains and pairwise independent variables}
\label{table3}
\begin{center}
\begin{tabular}{|c|c|} \hline \hline
Markov chains                    & Pairwise independent variables             \\ \hline\hline
Conditional pairwise independency   & Pairwise independency        \\
$I(X_1;X_3|X_2) = 0$      & $I(X_1;X_2) = 0$ \\
No $I_\text{priv}$ between $X_1$ and $X_3$   & No $I_\text{priv}$ between $X_1$ and $X_2$  \\
No synergistic information   & No shared information        \\ \hline\hline
\end{tabular}
\end{center}
\end{table}
%


\section{Entropy decomposition for the Gaussian case}
\label{sec:gaussian}

In this section we study the entropy-decomposition for the case where $(X_1,X_2,X_3)$ follow a multivariate Gaussian distribution. As the entropy is not affected by translation, we assume without loss of generality, that all the variables have zero mean. The covariance matrix is denoted as
\begin{equation}\label{eq:cov}
\Sigma
= \left(
\begin{array}{ccc}
\sigma^2_1 & \alpha\sigma_1\sigma_2 & \beta\sigma_1\sigma_3 \\
\alpha\sigma_1\sigma_2 & \sigma^2_2 & \gamma\sigma_2\sigma_3 \\
\beta\sigma_1\sigma_3 & \gamma\sigma_2\sigma_3 & \sigma^2_3
\end{array} \right)
\enspace,
\end{equation}
where $\sigma_i^2$ is the variance of $X_i$, $\alpha$ is the correlation between $X_1$ and $X_2$, $\beta$ is the correlation between $X_1$ and $X_3$ and $\gamma$ is the correlation between $X_2$ and $X_3$. The condition that the matrix $\Sigma$ should be positive semi-definite yields the following condition:
\begin{equation}\label{eq:posi}
1 + 2 \alpha \beta \gamma - \alpha^2 - \beta^2 - \gamma^2 \geq 0
\enspace.
\end{equation}
%


Unfortunately, Theorem~\ref{theo} implicitly states that Axioms (1)-(5) do not define a unique symmetrical information decomposition for Gaussian variables with an arbitrary covariance matrix. Nevertheless, there are some interesting properties of their shared and synergistic information, which are discussed in Sections~\ref{gaussiansyn} and \ref{understandingshared}. Then, Section~\ref{sec:gaussians4} presents one symmetrical information decomposition that is consistent with these properties.


\subsection{Understanding the synergistic information between Gaussians}
\label{gaussiansyn}

The simplistic structure of the joint p.d.f. of multivariate Gaussians, which is fully determined by mere second order statistics, could make one to think that these systems do not have synergistic information sharing. However, it can be shown that a multivariate Gaussian is the maximum entropy distribution for a given covariance matrix $\Sigma$. Hence, the discussion provided in Section~\ref{sec:PME} suggests that these distributions can indeed have non-zero information synergy, depending on the structure of the pairwise distributions, or equivalently, on the properties of $\Sigma$.

Moreover, it has been reported that synergistic phenomena are rather common among multivariate Gaussian variables \cite{barrett2015}. As a simple example, consider
\begin{equation}
X_1 = A+B, \quad X_2 = B, \quad X_3 = A,
\end{equation}
where $A$ and $B$ are independent Gaussians. Intuitively, it can be seen that although $X_2$ is useless by itself for predicting $X_3$, it can be used jointly with $X_1$ to remove the noise term $B$ and provide a perfect prediction. For refining this observation, let us consider a more general example where the variables have equal variances and $X_2$ and $X_3$ are independent (i.e. $\gamma=0$). Then, the optimal predictor of $X_3$ given $X_1$ is $\hat{X}_3^{X_1}=\alpha X_1$, the optimal predictor given $X_2$ is $\hat{X}_3^{X_2} = 0$, and the optimal predictor given both $X_1$ and $X_2$ is \cite{sayed2011}
\begin{equation}
\hat{X}_3^{X_1,X_2} = \frac{\beta}{1-\alpha^2} \left( X_1 - \alpha X_2 \right)
\enspace.
\end{equation}
Therefore, although $X_2$ is useless to predict $X_3$ by itself, it can be used for further improving the prediction given by $X_1$. Hence, all the information provided by $X_2$ is synergistic, as is useful only when combined with the information provided by $X_1$. Note that all these examples fall in the category of the systems considered in Section~\ref{sec:indep}.

\subsection{Understanding the shared information}
\label{understandingshared}

Let us start studying the information shared between two Gaussians. For this, let us consider a pair of zero-mean variables $(X_1,X_2)$ with unit variance and correlation $\alpha$. A suggestive way of expressing these variables is given by
\begin{equation}\label{eq:n=2}
X_1 = W_1 \pm W_{12}, \quad X_2 = W_2 \pm  W_{12} ,
\end{equation}
where $W_1$, $W_2$ and $W_{12}$ are independent centered Gaussian variables with variances $s_1^2 = s_2^2 = 1-|\alpha|$ and $s_{12}^2 = |\alpha|$, respectively. Note that the signs in \eqref{eq:n=2} can be set in order to achieve any desired sign for the covariance (as $\E{X_1X_2} = \pm \E{ W_{12}^2} = \pm s_{12}^2$). The mutual information is given by (see Appendix~\ref{sec:gaussianssss})

\begin{equation}\label{eq:mutualll}
I(X_1;X_2) = -(1/2) \log ( 1 - \alpha^2) = -(1/2) \log ( 1 - s_{12}^4)
\enspace,
\end{equation}
showing that it is directly related to the variance of the common term $W_{12}$.

For studying the shared information between three Gaussian variables, let us start considering a case where $\sigma_1^2 = \sigma_2^2 = \sigma_3^2=1$, $\alpha = \beta := \rho$ and $\gamma=0$. It can be seen that (c.f. Appendix~\ref{sec:gaussianssss})
\begin{equation}\label{eq:ssasdaS}
I(X_1;X_2;X_3) = \frac{1}{2} \log \frac { 1 - 2\rho^2}{ (1 - \rho^2)^2 }
\enspace.
\end{equation}
A direct evaluation shows that \eqref{eq:ssasdaS} is non-positive\footnote{This is consistent with the fact that $X_2$ and $X_3$ are pairwise independent, and hence due to \eqref{eqlema4} one has that $0 \leq \Syn{X_1;X_2;X_3} = - I(X_1;X_2;X_3)$.} for all $\rho$ with $|\rho| < 1/\sqrt{2} $ (note that $|\rho|$ cannot be larger that $1/\sqrt{2}$ because of condition \eqref{eq:posi}). Therefore, following the discussion related to \eqref{eq:g3}, this system has no shared information for all $\rho$ and has zero synergistic information only for $\rho = 0$. In contrast, let us now consider a case where $\alpha=\beta=\gamma :=\rho>0$, for which
\begin{equation}
I(X_1;X_2;X_3) = \frac{1}{2} \log \frac { 1 + 2 \rho^3 - 3\rho^2}{ (1 - \rho^2)^3 }
\enspace.
\end{equation}
A direct evaluation shows that, in contrast to \eqref{eq:ssasdaS}, the co-information in this case is non-negative, showing that the system is dominated by shared information for all $\rho\neq 0$.

The previous discussion suggests that the shared information depends on the smallest of the correlation coefficients. An interesting approach to understand this fact can be found in \cite{barrett2015}, where the predictability among Gaussians is discussed. In this work, the authors note that from the point of view of $X_3$ both $X_1$ and $X_2$ are able to decompose the target in a predictable and an unpredictable portion: $X_3 = \hat{X}_3 + E$. In this sense, both predictors achieve the same effect although with a different efficiency, which is determined by their correlation coefficient. As a consequence of this, the predictor that is less correlated with the target does not provide unique predictability and hence its contribution is entirely redundant. This motivates the following redundant predictability measure:
\begin{equation}
\dr{X_1X_2}{X_3} \coloneqq \min\{ I(X_1;X_3), I(X_2;X_3) \}.\label{eqwewe}
\end{equation}

\subsection{Shared, private and synergistic information for Gaussian variables}
\label{sec:gaussians4}

Let us use the intuitions developed in the previous section for building a symmetrical information decomposition. For this, we use the decomposition given by the following Lemma (whose proof is presented in Appendix~\ref{leamaamamama}).
\begin{lemma}\label{lemmagaus}
Let $(X_1,X_2,X_3)$ follow a multivariate Gaussian distribution with zero mean and covariance matrix $\Sigma$ with $\alpha \geq \beta \geq \gamma \geq 0$. Then
\begin{align}
\frac{X_1}{\sigma_1} &= s_{123} W_{123} + s_{12} W_{12} + s_{13} W_{13} + s_1 W_1 \\
\frac{X_2}{\sigma_2} &= s_{123} W_{123} + s_{12} W_{12} +  s_{2} W_2 \\
\frac{X_3}{\sigma_3} &= s_{123} W_{123} + s_{13} W_{13} + s_{3} W_3
\end{align}
where $W_{123}, W_{12}, W_{13},W_1,W_2$ and $W_3$ are independent standard Gaussians and $s_{123}, s_{12},s_{13},s_1,s_2$ and $s_3$ are given by
\begin{align}
s_{123}&= \sqrt{\gamma},\quad s_{12} = \sqrt{\alpha-\gamma},\quad s_{13} = \sqrt{\beta-\gamma},\nonumber\\
s_1&=\sqrt{1-\alpha-\beta+\gamma},\quad s_2=\sqrt{1-\alpha},\quad s_3 = \sqrt{1-\beta}. \label{paramss}
\end{align}
\end{lemma}

It is natural to relate $s_{123}$ with the shared information, $s_{12}$ and $s_{13}$ with the private information and $s_1$, $s_2$ and $s_3$ with the exclusive terms. Note that the decomposition presented in Lemma~\ref{lemmagaus} is unique in not requiring a private component between the two less correlated variables ---i.e. a term $W_{23}$. Hence, based on Lemma~\ref{lemmagaus} and \eqref{eq:mutualll}, we propose the following symmetric information decomposition for Gaussians:
\begin{align}
I_\cap(X_1;X_2;X_3) &= -\frac{1}{2} \log ( 1 - \min\{\alpha^2,\beta^2,\gamma^2\} )\enspace,  \label{eqeqeqe}\\
\Un{X_1;X_2|X_3} &= I(X_1;X_2) -I_\cap(X_1;X_2;X_3) \\
&= \frac{1}{2} \log \frac{  1 - \min\{\alpha^2,\beta^2,\gamma^2\}  }{ 1 - \alpha^2}\enspace, \label{eq:asdasdasdgafg}\\
I_\text{S}(X_1;X_2;X_3) &= I(X_1;X_2|X_3) - \Un{X_1;X_2|X_3}\\
&= \frac{1}{2} \log \frac {( 1 - \alpha^2)( 1 - \beta^2)( 1 - \gamma^2)}{ (1+2\alpha\beta\gamma -\alpha^2 - \beta^2 - \gamma^2)( 1 - \min\{\alpha^2,\beta^2,\gamma^2\})} \label{eq:asdqergqregr} \enspace.
\end{align}

First, note that the above shared information coincides with what was expected from Lemma~\ref{lemmagaus}, as for the general case $s_{123}^2=\min\{ |\alpha|, |\beta|, |\gamma|\}$. Also, \eqref{eq:asdasdasdgafg} is consistent with the fact that the two less correlated Gaussians share no private information. Moreover, by comparing \eqref{eq:asdqergqregr} and \eqref{eq:condGauss}, it can be seen that if $X_1$ and $X_2$ are the less correlated variables then the synergistic information can be expressed as $\Syn{X_1;X_2;X_3} = I(X_1;X_2|X_3)$, which for the particular case of $\alpha=0$ confirms \eqref{eqlema4}. This in turn also shows that, for the particular case of Gaussians variables, forming a Markov chain is a necessary and sufficient condition for having zero information synergy\footnote{For the case of $\alpha \geq \beta \geq \gamma$, a direct calculation shows that $I(X_1;X_2|X_3) = 0$ is equivalent to $\gamma=\alpha\beta$.}. 

Finally, by noting that \eqref{eqeqeqe} can also be expressed as
\begin{equation}\label{eqs"as}
I_\cap(X_1;X_2;X_3) = \min\{ I(X_1;X_2), I(X_2;X_3), I(X_1;X_3) \}
\enspace,
\end{equation}
it can be seen that our definition of shared information corresponds to the canonical symmetrization of \eqref{eqwewe} as discussed in Lemma~\ref{lemaaaa}. In contrast with \eqref{eqwewe}, \eqref{eqs"as} states that there cannot be information shared by the three components of the system if two of them are pairwise independent. Therefore, the magnitude of the shared information is governed by the lowest correlation coefficient of the whole system, being upper-bounded by any of the redundant predictability terms.

To close this section, let us note that \eqref{eqs"as} corresponds to the upper bound provided by \eqref{eq:bounds}, which means that multivariate Gaussians have a maximal shared information. This is complementary to the fact that, because of being a maximum entropy distribution, they also have the smallest amount of synergy that is compatible with the corresponding second order statistics.


\section{Applications to Network Information Theory}
\label{sec:4}

In this section we use the framework presented in Section~\ref{sec:3} to analyze four fundamental scenarios in network information theory \cite{el2011network}. Our goal is to illustrate how the framework can be used to build new intuitions over these well-known optimal information-theoretic strategies. The application of the framework to scenarios with open problems is left for future work.

In the following, Section~\ref{sec:SW_N=3} uses the general framework to analyze the Slepian-Wolf coding for three sources, which is a fundamental result in the literature of distributed source compression. Then, Section~\ref{sec:mac} applies the results of Section~\ref{sec:indep} to the multiple access channel, which is one of the fundamental settings in multiuser information theory. Section~\ref{sec:wiretap} uses the results related to Markov chains from Section~\ref{sec:V} to the wiretap channel, which constitutes one of the main models of information-theoretic~secrecy. Finally, Section~\ref{sec:GBCC} uses results from Section~\ref{sec:gaussian} to study fundamental limits of public or private broadcast transmissions over Gaussian channels.

\subsection{Slepian-Wolf coding}\label{sec:SW_N=3}

The Slepian-Wolf coding gives lower bounds for the data rates that are required to transfer the information contained in various data sources. Let us denote as $R_k$ the data rate of the $k$-th source and define $\tilde{R}_k = R_k - H(X_k|\mathbf{X}_k^\text{c})$ as the extra data rate that each source has above their own exclusive information (c.f. Section~\ref{sec:2b}). Then, in the case of two sources $X_1$ and $X_2$, the well-known Slepian-Wolf bounds can be re-written as $\tilde{R}_1 \geq 0$, $\tilde{R}_2 \geq 0$,  and $\tilde{R}_1 + \tilde{R}_2 \geq I(X_1;X_2)$ \cite[Section 10.3]{el2011network}. The last inequality states that $I(X_1;X_2)$ corresponds to shared information that can be transmitted by any of the two sources.

Let us consider now the case of three sources, and denote $R_S=\Syn{X_1;X_2;X_3}$. The Slepian-Wolf bounds provide seven inequalities \cite[Section 10.5]{el2011network}, which can be re-written as
\begin{align}
\tilde{R}_i & \geq 0,\; i\in\{1,2,3\} \\
\tilde{R}_i + \tilde{R}_j & \geq \Un{X_i;X_j|X_k} + R_S \;\,\text{for} \; i,j,k\in \{1,2,3\}, i<j \label{eq:11}\\
\tilde{R}_1 + \tilde{R}_2 + \tilde{R}_3 &\geq \Delta H_{(2)} + \Delta H_{(3)} \label{eq:22}
\end{align}
Above, \eqref{eq:22} states that the DTC needs to be accounted by the extra rate of the sources, and \eqref{eq:11} that every pair needs to to take care of their private information. Interestingly, due to \eqref{eq:3} the shared information needs to be included in only one of the rates, while the synergistic information needs to be included in at least two. For example, one possible solution that is consistent with these bounds is $\tilde{R}_1= \Red{X_1;X_2;X_3} + \Un{X_1;X_2|X_3} + \Un{X_1;X_3|X_3} + \Syn{X_1;X_2;X_3}$, $\tilde{R}_2 = \Un{X_2;X_3|X_1} + \Syn{X_1;X_2;X_3}$ and $\tilde{R}_3=0$.

%
%

\subsection{Multiple Access Channel}
\label{sec:mac}

Let us consider a multiple access channel, where two pairwise independent transmitters send $X_1$ and $X_2$ and a receiver gets $X_3$ as shown in Fig.~\ref{fig:MAC}. It is well-known that, for a given distribution $(X_1,X_2)\sim p(x_1)p(x_2)$, the achievable transmission rates $R_1$ and $R_2$ satisfy the constrains \cite[Section 4.5]{el2011network}
\begin{equation}
 R_1 \leq I(X_1;X_3|X_2),\quad
R_2 \leq I(X_2;X_3|X_1),\quad
R_1+R_2 \leq I(X_1,X_2;X_3).
\end{equation}
As the transmitted random variables are pairwise independent, one can apply the results of Section~\ref{sec:indep}. Therefore, there is no shared information and $I_\text{S}(X_1;X_2;X_3) =I(X_1;X_3|X_2) - I(X_1;X_3)$. Let us introduce a shorthand notation for the remaining terms : $C_1 = \Un{X_1;X_3|X_2} =I(X_1;X_3)$, $C_2 = \Un{X_2;X_3|X_1}=I(X_2;X_3)$ and $C_\text{S} = \Syn{X_1;X_2;X_3}$. Then, one can re-write the bounds for the transmission rates as
\begin{equation}\label{eq:asasasdas}
 R_1\leq C_1 + C_\text{S},\quad R_2\leq C_2 + C_\text{S} \quad \text{and} \quad R_1+R_2 \leq C_1 + C_2 + C_\text{S}.
\end{equation}
From this, it is clear that while each transmitter has a private portion of the channel with capacity $C_1$ or $C_2$, their interaction creates \emph{synergistically} extra capacity $C_\text{S}$ that corresponds to what can be actually shared.
\begin{figure}[ht]

                \centering
                \includegraphics[width=1\columnwidth]{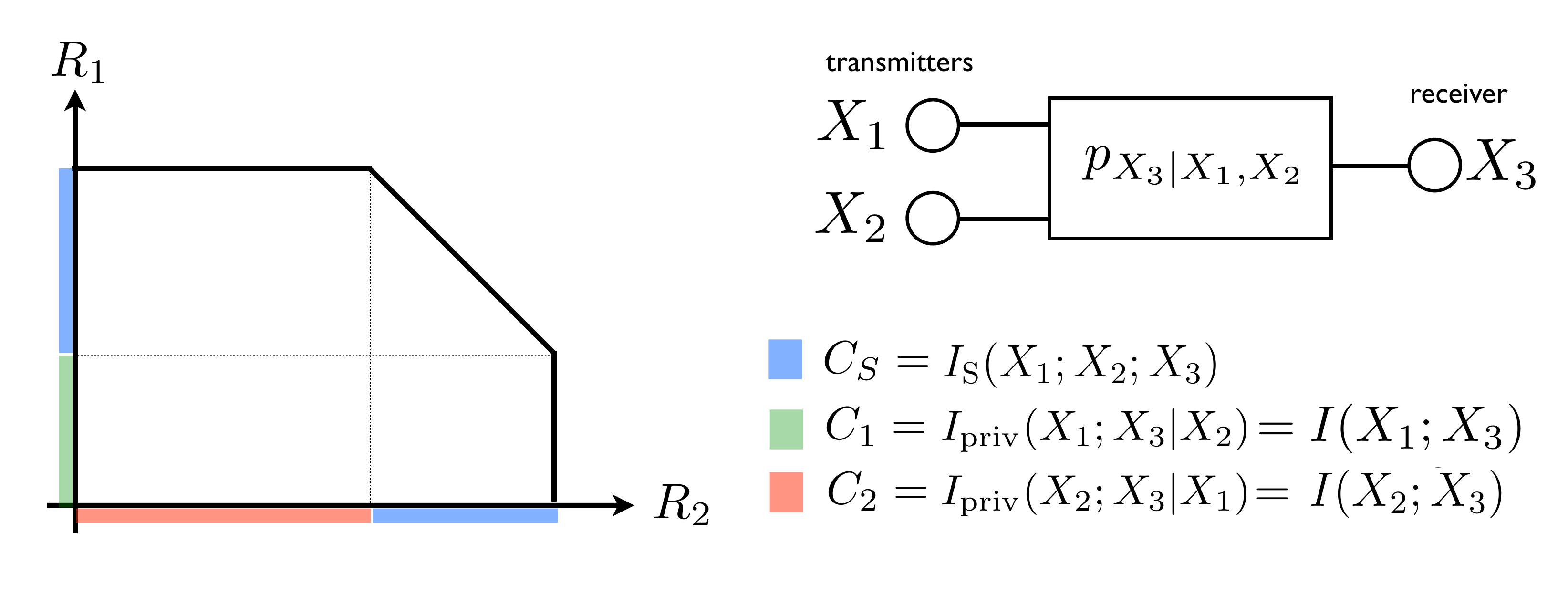}
                \caption{Capacity region of the \textit{Multiple Access Channel}, which represents the possible data-rates that two transmitters can use for transferring information to one receiver.}
                \label{fig:MAC}

\end{figure}


\subsection{Degraded Wiretap Channel}
\label{sec:wiretap}

Consider a communication system with an eavesdropper (shown in Fig.~\ref{fig:Wiretap}), where the transmitter sends  $X_1$, the intended receiver gets $X_2$ and the eavesdropper receives $X_3$. For simplicity of the exposition, let us consider the case where the eavesdropper get only a degraded copy of the signal received by the intended receiver, i.e. that $X_1 - X_2 - X_3$ form a Markov chain. Using the results of Section~\ref{sec:markov}, one can see that in this case there is no synergistic but only shared and private information between $X_1$, $X_2$ and $X_3$.
\begin{figure}[ht]

                \centering
                \includegraphics[width=1\columnwidth]{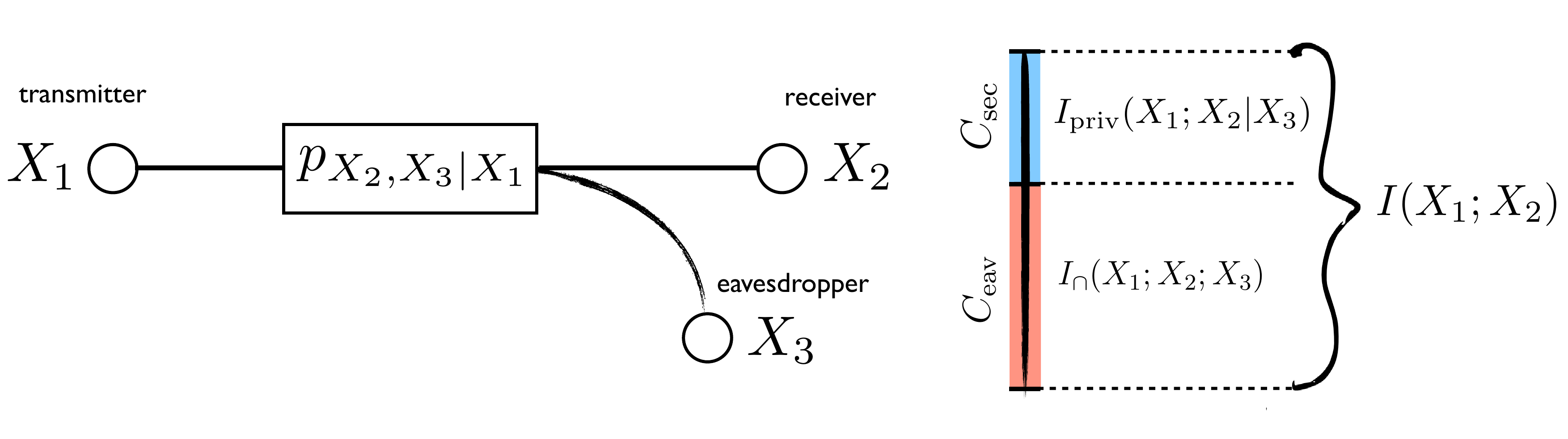}
                \caption{The rate of secure information transfer, $C_\text{sec}$, is the portion of the mutual information that can be used while providing perfect confidentiality with respect to the eavesdropper.\vspace{-0.1in}}
                \label{fig:Wiretap}

\end{figure}

In this scenario, it is known that for a given input distribution $p_{X_1}$ the rate of secure communication that can be achieved is upper bounded by \cite[Section 3.4]{bloch2011}
\begin{equation}
C_\text{sec} = I(X_1;X_2) - I(X_1;X_3) = \Un{X_1;X_2|X_3},
\end{equation}
which is precisely the private information sharing between $X_1$ and $X_2$. Also, as intuition would suggest, the eavesdropping capacity is equal to the shared information between the three variables:
\begin{equation}
 C_\text{eav} = I(X_1;X_2) - C_\text{sec} = I(X_1;X_3)=\Red{X_1;X_2;X_3}.
\end{equation}

\subsection{Gaussian Broadcast Channel}
\label{sec:GBCC}
Let us consider a Gaussian Broadcast Channel, where a transmitter sends a Gaussian signal $X_1$ that is received as $X_2$ and $X_3$ by two receivers. Assuming that all these variables jointly Gaussian with zero mean and covariance matrix as given by \eqref{eq:cov}, the transmitter can broadcast a public message, intended for both users, at a maximum rate $C_\text{pub}$ given by \cite[Section 5.1]{bloch2011}
\begin{equation}
C_\text{pub} = \min\{ I(X_1;X_2), I(X_1;X_3)\} = \dr{X_2X_3}{X_1}
\enspace,
\end{equation}
where the redundant predictability,  $\dr{X_2X_3}{X_1}$, between Gaussian variables is as defined in \eqref{eqwewe}. On the other hand, if the transmitter wants to send a private (confidential) message to receiver~1, the corresponding maximum rate $C_\text{priv}$ that can be achieved in this case is given by
\begin{equation}
C_\text{priv} = [ I(X_1;X_2) - I(X_1;X_3) ]^+ = I(X_1;X_2) - \dr{X_2X_3}{X_1} = \du{X_1}{X_2|X_3}
\enspace,
\end{equation}
where the last equality follows from Axiom (2).

Interestingly, the predictability measures prove to be better suited to describe the communication limits in the above scenario that their symmetrical counterparts. In effect, using the shared information would have underestimated the public capacity (c.f. Section~\ref{sec:gaussians4}). This opens the question whether or not directed measures could be better suited for studying certain communication systems, compared to their symmetrized counterparts. Even though a definite answer to this question might not be straightforward, we hope that future research will provide more evidence and a better understanding of this issue.


\section{Conclusions}
\label{sec:conclusions}

In this work we propose an axiomatic framework for studying the interdependencies that can exist between multiple random variables as different modes of information sharing. The framework is based on a symmetric notion of information that refers to  properties of the system as a whole. We showed that, in contrast to predictability-based decompositions, all the information terms of the proposed decomposition have unique expressions for Markov chains and for the case where two variables are pairwise independent. We also analyzed the cases of pairwise maximum entropy (PME) distributions and multivariate Gaussian variables. Finally, we illustrated the application of the framework by using it to develop a more intuitive understanding of the optimal information-theoretic strategies in several fundamental communication scenarios.

The key insight that this framework provides is that although there is only one way in which information can be shared between two random variables, there are two essentially different ways of sharing between three. One of these ways is a simple extension of the pairwise dependency, where information is shared redundantly and hence any of the variables can be used to predict any other. The second way leads to the counter-intuitive notion of synergistic information sharing, where the information is shared in a way that the statistical dependency is destroyed if any of the variables is removed; hence, the structure exists in the whole but not in any of the parts. Information synergy has therefore been commonly related to statistical structures that exist only in the joint p.d.f. and not in low-order marginals. Interestingly, although we showed that indeed PME distributions posses the minimal information synergy that is allowed by their pairwise marginals, this minimum can be strictly positive.

Therefore, there exists a connection between pairwise marginals and synergistic information sharing that is still to be further clarified. In fact, this phenomenon is  related to the difference between the TC and the DTC, which is rooted in the fact that the information sharing modes and the marginal structure of the p.d.f. are, although somehow related, intrinsically different. This important distinction has been represented in our framework by the sequence of internal and external entropies. This new unifying picture for the entropy, negentropy, TC and DTC has shed new light in the understanding of high-order interdependencies, whose consequences have only begun to be explored.


\appendices

\section{Proof of Lemma~\ref{lemma1}}
\label{ap:lemma1}

\begin{proof}
Let us assume that $\dr{X_1X_2}{Y}$ and $\du{X_1}{Y|X_2}=I(X_1;Y) - \dr{X_1X_2}{Y}$ satisfy Axioms (1)--(3). Then,
\begin{align}
I(X_1;Y) &\geq I(X_1;Y) - \du{X_1}{Y|X_2} \\
&= \dr{X_1X_2}{Y}\\
&= I(X_2;Y) - \du{X_2}{Y|X_1} \leq I(X_2;Y)
\end{align}
where the inequalities are a consequence of the non-negativity of $\du{X_1}{Y|X_2}$ and the third equality is due to the weak symmetry of the redundant predictability. For proving the lower bound, first notice that Axiom (2) can be re-written as
\begin{equation}
I(X_1X_2;Y) \geq I(X_1;Y) + I(X_2;Y) - \dr{X_1X_2}{Y}.
\end{equation}
The lower bound follows considering the non-negativity of $\dr{X_1X_2}{Y}$ and by noting that $ I(X_1;Y) + I(X_2;Y) - I(X_1X_2;Y) = I(X_1;X_2;Y)$.

The proof of the converse is direct, and left as an exercise to the reader.
\end{proof}


\section{Proof of the consistency of Axiom (3)}
\label{app:1}

Let us show that $\min\{ I(X_1;X_2), I(X_1;X_2)\} \geq I(X_1;X_2;X_3)$, showing that the bounds defined by Axiom (3) always can be satisfied. For this, let us assume that the variables are ordered in a way such that $I(X_1;X_2) = \min\{ I(X_1;X_2), I(X_2;X_3), I(X_3;X_1) \}$ holds. Then, as one can express $I(X_1;X_2;X_3) = I(X_1,X_2) -  I(X_1,X_2|X_3)$, it is direct to show that
\begin{align}
\min\{ I(X_1;X_2), I(X_1;X_2)\} - I(X_1;X_2;X_3) &\geq I(X_1;X_2) -  I(X_1;X_2;X_3) \\
&= I(X_1;X_2|X_3)\\
&\geq 0
\enspace,
\end{align}
from where the desired result follows.


\section{Proof of Lemma~\ref{lemma1}}
\label{app:2}

\begin{proof}
The symmetry of $\Red{X_1;X_2;X_3}$ can be directly verified from its definition. The weak symmetry of $\Un{X_1;X_3|X_2}$ can be shown as follows:
\begin{align}\label{eq:ssss}
\Un{X_3;X_1|X_2} &= I(X_3;X_1) - \Red{X_3; X_2;X_1} \\
&= I(X_1;X_3) - \Red{X_1; X_2;X_3} \\
&= \Un{X_1;X_3|X_2}
\enspace.
\end{align}
The symmetry of $\Syn{X_1;X_2;X_3}$ with respect to $X_1$ and $X_3$ follows directly from its definition, the weak symmetry of $I(X_1;X_3|X_2)$ and the strong symmetry of $\Red{X_1;X_2;X_3}$. The symmetry with respect to $X_1$ and $X_2$ can be shown using the definition of $\Syn{X_1;X_2;X_3}$ and the strong symmetry of $\Red{X_1;X_2;X_3}$ and the co-information $I(X_1;X_2;X_3)$ as follows:
\begin{align}
\Syn{X_2;X_1;X_3} &= I(X_2;X_3|X_1) - [ I(X_2;X_3) - \Red{X_2;X_1;X_3}] \\
&=  I(X_1;X_2;X_3) + \Red{X_2;X_1;X_3}\\
&=  I(X_1;X_3|X_2) -  I(X_1;X_3) + \Red{X_1;X_2;X_3}\\
&= \Syn{X_1;X_2;X_3}
\enspace.
\end{align}

The bounds for $\Red{X_1;X_2;X_3}$, $\Un{X_1;X_2;X_3}$ and $\Syn{X_1;X_2;X_3}$ follow directly from the definition of these quantities and Axiom (3). Finally, $d)$ is proven directly using those definitions, and the fact that the mutual information depend only on the pairwise marginals, while the conditional mutual information depends on the full p.d.f.
\end{proof}


\section{Useful facts about Gaussians}
\label{sec:gaussianssss}
Here we list some useful expressions for Gaussian variables:
\begin{align}
I(X_1;X_2) &= \frac{1}{2} \log \frac{1}{1-\alpha^2} \\
&= \frac{1}{2} \log \frac{\sigma^2}{| \Sigma_{12}|} \enspace,\\
I(X_1;X_2,X_3) &= \frac{1}{2} \log  \frac{1-\gamma^2}{1 + 2\alpha\beta\gamma - \alpha^2 - \beta^2 - \gamma^2} \\
&= \frac{1}{2} \log \frac{ |\Sigma_{23}|}{ |\Sigma|} \enspace,\\
I(X_1;X_2|X_3) &=  \frac{1}{2} \log  \frac {(1-\beta^2) (1-\gamma^2)}{1 + 2\alpha\beta\gamma - \alpha^2 - \beta^2 - \gamma^2} \label{eq:condGauss}\\
&= \frac{1}{2} \log \frac{ |\Sigma_{13}\Sigma_{23}|}{ |\Sigma|} \enspace,\\
I(X_1;X_2;X_3) &=  \frac{1}{2} \log  \frac{1 + 2\alpha\beta\gamma - \alpha^2 - \beta^2 - \gamma^2 }{(1-\alpha^2) (1-\beta^2) (1-\gamma^2)} \\
&= \frac{1}{2} \log \frac { |\Sigma|} { |\Sigma_{12}\Sigma_{13}\Sigma_{23}|}\enspace,
\end{align}
where $|\Delta|$ is a matrix determinant, and
\begin{equation}
\Sigma_{12} = \left(
\begin{array}{cc}
\sigma^2 & \alpha\sigma^2 \\
\alpha\sigma^2 & \sigma^2 
\end{array} \right)
\qquad
\Sigma_{13} = \left(
\begin{array}{cc}
\sigma^2 & \beta\sigma^2 \\
\beta\sigma^2 & \sigma^2 
\end{array} \right)
\qquad
\Sigma_{23} = \left(
\begin{array}{cc}
\sigma^2 & \gamma\sigma^2 \\
\gamma\sigma^2 & \sigma^2 
\end{array} \right)
\enspace.
\end{equation}
%


\section{Proof of Lemma~\ref{lemmagaus}}
\label{leamaamamama}

\begin{proof}
Consider the following random variables
\begin{align}
Y_1 &= \sigma_1 ( s_{123} W_{123} + s_{12} W_{12} + s_{13} W_{13} + s_1 W_1 )\\
Y_2 &= \sigma_2 (s_{123} W_{123} + s_{12} W_{12} +  s_{2} W_2) \\
Y_3 &= \sigma_3 (s_{123} W_{123} + s_{13} W_{13} + s_{3} W_3)
\end{align}
where $W_{123}, W_{12}, W_{13},W_1,W_2$ and $W_3$ are independent standard Gaussians and the parameters $s_{123}, s_{12},s_{13},s_1,s_2$ and $s_3$ as defined in \eqref{paramss}.
Then, is direct to check that $\bold{Y}=(Y_1,Y_2,Y_3)$ is a multivariate Gaussian variable with zero mean and covariance matrix $\Sigma_\bold{Y}$ equal to \eqref{eq:cov}. Therefore, $(Y_1,Y_2,Y_3)$ and $(X_1,X_2,X_3)$ have the same statistics, which proves the desired result.
\end{proof}

\section*{Acknowledgments}

We want to thank David Krakauer and Jessica Flack for providing the inspiration for this research. We also thank Bryan Daniels, Michael Gastpar, Bernhard Geiger, Vigil Griffith and Martin Ugarte for helpful discussions. This work was partially supported by a grant to the Santa Fe Institute for the study of complexity and by the U.S. Army Research Laboratory and the U.S. Army Research Office under contract number W911NF-13-1-0340. FR would also like to acknowledge the support of the F+ fellowship from KU Leuven and the SBO project ”SINS”, funded by the Agency for Innovation by Science and Technology IWT, Belgium.

\bibliographystyle{IEEEtran}
\bibliography{\bibpath/library}





\ifCLASSOPTIONjournal 
\fi

\end{document}